\documentclass{article}
\usepackage{macros}

\title{Novel Encodings of Homology, Cohomology, and Characteristic Classes}
\author{Itai Maimon}
\affil{Department of Mathematics \\
University of California San Diego \\
La Jolla, CA 92093-0112 (USA) \\
\{\tt imaimon\}@ucsd.edu}
\date{\today}

\begin{document}

\maketitle
\begin{abstract}
Topological quantum error-correcting codes (QECC) encode a variety of topological invariants in their code space \cite{kitaev2003fault} \cite{levin2005string}. A classic structure that has not been encoded directly is that of obstruction classes of a fiber bundle, such as the Chern or Euler class. Here, we construct and analyze extensions of toric codes. We then analyze the topological structure of their errors and finally construct a novel code using these errors to encode the obstruction class to a fiber bundle. In so doing, we construct an encoding of characteristic classes such as the Chern and Pontryagin class in topological QECC. An example of the Euler class of $S^2$ is constructed explicitly. 

$\,$

$\,$

\end{abstract}

\tableofcontents
\newpage

\section{Introduction: Explicitly Encoding Higher Homologies}

For a manifold $M$ and a finitely generated abelian group $A$, we will start by explicitly defining an encoding of any homology and cohomology of $M$ over $A$, $H_*(M, A), H^*(M, A)$. To do so, we will first prove the following lemmas \cite{kitaev2003fault} :
\begin{lem}[Higher Homologies over $\Z_2$]
For any integer $k\geq 0$ and CW-complex $T$, there exists a stabilizer QECC in which $H_k(T,\Z_2)$ is the code-space.
\end{lem} 

\begin{proof}
We first define $T_i$ as the set of $i$-cells of $T$, $\partial$ as the boundary map on chains, and for some chain $v$, $\{\partial v\}$ as the set of cells on the boundary of $v$. On $T$ place a qubit on every $k$-cell; which constructs the Hilbert space $H=\bigotimes_{i\in T_k} \C^2$, where $T_k$ is the set of $k$-cells in $T$. As shorthand, an operator acting on a $k$-cell refers to an operator acting on the qubit associated with that $k$-cell, tensored with the identity on the rest of the Hilbert space. For each $(k-1)$-cell, $\alpha$, we construct the operator which acts by $\sigma_z$ on the the $k$-cells with boundary containing $\alpha$. These cells are referred to as the coboundary of $\alpha$, symbolically $\{\partial^\top \alpha\}$.

Again, explicitly, this is the operator:

\begin{equation}V_\alpha=\bigotimes_{i \in \{\partial^\top \alpha\}} \sigma_z \otimes \bigotimes_{i \in T_k \setminus \{\partial^\top \alpha\}}I\end{equation}

Similarly, for each $(k+1)$-cell, $\gamma$, construct the operator that acts as $\sigma_x$ on all $k$-cells on the boundary of $\gamma$. 
Explicitly, this is the operator:

\begin{equation}P_\gamma=\bigotimes_{i\in \{\partial \gamma\}} \sigma_x \otimes \bigotimes_{i \in T_k \setminus \{\partial \gamma\}} I\end{equation}

In order to simplify notation, define $A_r$ as the operator $A$ applied to the $k$-cell $r$ tensored to the identity acting on all other $k$-cells. For instance:

\begin{equation}V_\alpha=\prod_{i \in \{\partial^\top \alpha\}} \sigma_{z,i} \ , \quad P_\gamma=\prod_{i\in \{\partial \gamma\}} \sigma_{x,i}
\end{equation}

In a valid stabilizer code, these operators must commute. For any choice of cells, any two such operators, $X$ and $Y$, either do not intersect on any $k$-cells or intersect on exactly two $k$-cells. If they do not intersect, they commute as they act on different tensor parts. If they intersect on two cells, $r,s$, then the commutator of the two operators is the tensor product of the commutators on each cell, e.g., for two operators $X$ and $Y$:
\begin{equation}[X,Y]={X}^{-1}{Y}^{-1}XY=X_r^{-1}Y_r^{-1}X_rY_r\otimes X_s^{-1}Y_s^{-1}X_sY_s=[X_r,Y_r]\otimes [X_s,Y_s]\end{equation}

As these operators act on all cells, which individually commute if they are the same and anti-commute otherwise, the commutator is either $-1\otimes -1\otimes I=I$ or $1\otimes 1\otimes I= I$. So, we have shown that these operators commute. Given these operators, consider the Hamiltonian acting on the space $H$ of the form: 

\begin{equation}H_{Toric}=\sum_{\alpha \in T_{k-1}} (1-V_\alpha) +\sum_{\gamma \in T_{k+1}} (1-P_\gamma)\end{equation}

We define the code-space as the lowest eigenvalue eigenstate of this operator, which is equivalent to defining it as the 1-eigenspace of the stabilizer operations $V_\alpha, P_\gamma$. We will show that this code-space corresponds to $H_k(T,\Z_2)$.

Each vector in the computational basis, i.e., $\ket{0001110101010\dots}$, is associated with a choice of $0$ or $1$ to each $k$-cell. This association defines an isomorphism from $H$, with its computational basis, to the group algebra over $\C$ of the $k$-chains of $T$ over $\Z_2$, $\C[\bigoplus_{i\in T_k} \Z_2]$, with its standard basis. This isomorphism is given by linearly extending the following bijection, $i$ of basis vectors:

\begin{equation}i(\ket{l_1,l_2\dots l_r})= v_{_{l_1 e_1+l_2 e_2 \dots + l_r e_r}}\end{equation}

We will first show that the common 1-eigenspace of the collection of the $V$-operators is the subspace of $H$ generated by basis vectors corresponding to $k$-cycles. Any group element, $\beta$, that is not a cycle has at least one $(k-1)$-cell, $\alpha$, in its boundary, so the eigenvalue of $V_{\alpha}$ on $\ket{\beta}$ is $-1$. The intersection of the $1$-eigenspace of all $V_\alpha$ operators can then only contain linear combinations of states corresponding to cycles. On the other hand, for any $k$-cycle $\beta$, each $V_\alpha$ operator acting on $\ket{\beta}$ multiplies the overall state by $(-1)^{|\{\partial^\top \alpha\}\cap \beta |}$. As $\beta$ is a $k$-cycle over $\Z_2$, $\{\partial^\top \alpha \} \cap \beta$ is an even number of cells, and so all cycles are contained in each $V_\alpha$ operator's 1-eigenspace. The common $1$-eigenspace of $V$-operators must then be generated by basis vectors corresponding to $k$-cycles.

Continuing with the construction, we will show that the common $1$-eigenspace of the collection of the $P$-operators is the subspace of $H$ spanned by elements of the form: 
\begin{equation}h_{[\beta]}=\sum_{\beta_i \in [\beta]} \ket{\beta_i}\end{equation}

Where $\beta_i$ is a $k$-chain in the equivalence class of homologous chains $[\beta]$. Let $k$-chains $\beta_1,\beta_2$ both be in the same equivalence class $[\beta]$, with corresponding states $\ket{b_1}$, and $\ket{b_2}$. These $k$-chains must differ by a boundary of a $(k+1)$-chain $\gamma$. The action of a $P$ operator corresponding to one of the $(k+1)$-cells, $\gamma_1$, with non-zero coefficient in the sum $\gamma$, takes the basis vector $\ket{b_1}$ to a basis vector, which we can denote $\ket{b_{(1,\gamma_1)}}$ and vice versa. This implies that the common 1-eigenspace of the collection of the $P$-operators must have the coefficients of $\ket{b_1}$ and $\ket{b_{(1,\gamma_1)}}$ equivalent. Repeating this for each $(k+1)$-cell in the $(k+1)$-chain $\gamma$ determines that the coefficient of $\ket{b_1}$ and $\ket{b_2}$ must be the same. Therefore, the $1$-eigenspace of the $P$ operators must be contained in the subspace spanned by the vectors $h_{[\beta]}$. On the other hand, for any $\beta$, all $P_\gamma$ operators act on $h_{[\beta]}$ by permuting the elements in the sum, i.e., 
\begin{equation}P_{\gamma}(\ket{\beta}+\ket{\partial \gamma +\beta})=\ket{\partial \gamma +\beta}+\ket{2\partial \gamma+\beta}=\ket{\beta}+\ket{\partial \gamma +\beta}\end{equation}

More generally, $P_\gamma(h_{[\beta]})=h_{[\beta]}$, and so the intersection of the 1-eigenspaces of all the $P_\gamma$ operators is the linear combinations of $h_{[\beta]}$ for all equivalence classes $[\beta]$. We have then shown that the intersection of the $1$-eigenstate of both $P$-type and $V$-type stabilizer operators is the intersection of the vector spaces generated by cycles and the vector space generated by vectors $h_{[b]}$. Therefore, the code-space is the linear combination of elements of the form $h_{[\beta]}$ where $\beta$ is a $k$-cycle. As any chain homologous to a cycle is a cycle, this is well-defined, and as these $[\beta]$ are homology classes over $\Z_2$, the ground state is isomorphic to $\C[H_k(T,\Z_2)]$
\end{proof}

We can now complete the argument by extending this construction for coefficients in any finitely generated abelian group.

\begin{lem}[Higher Homologies over Finitely Generated Abelian Groups]
For any integer $k\geq 0$, a finitely generated abelian group, $A$, and a finite CW-complex, $T$, a stabilizer QECC exists in which the ground space isomorphic to $C[H_k(T, A)]$.
\end{lem}

\begin{proof}
Before proving the above over any abelian group, we consider A=$\Z_{d}$.
As before, place qudits on every $k$-cell, more precisely constructing the Hilbert space: 

\begin{equation}H=\bigotimes_{i\in T_k} \C^d\end{equation}

This is constructed so that $H$ is isomorphic to the group algebra over $\C$ of the $k$-chains of $T$ over $\Z_d$, $\C[\bigoplus_{i\in T_k} \Z_d]$. As before, the isomorphism is given by:
\begin{equation}i(\ket{l_1,l_2\dots l_r})= v_{_{l_1 e_1+l_2 e_2 \dots + l_r e_r}}\end{equation}

As the operators no longer act on qubits, we must use different stabilizer operators. For $\zeta_d$ a $d$th root of unity, we replace the Pauli $Z$ and $X$ with the $Z_{d}$ and $X_{d}$ operators below \cite{HostensDehaeneDeMoor2005_StabilizerCliffordQudits}. 

\begin{equation}Z_{d}=
\begin{bmatrix}
1 & 0 & 0 & \cdots & 0 \\ 
0 & \zeta_{d} & 0 &\cdots & 0\\
0 & 0 & \zeta_d^2 & \cdots & 0 \\
\vdots & \ddots & \ddots & \ddots & \vdots\\ 
0 & \cdots & \cdots &\cdots & \zeta_d^{-1} \\
\end{bmatrix}
\quad 
X_{d}=
\begin{bmatrix}
0 & 0 & 0 & \cdots & 0 &  1 \\ 
1 & 0 & 0 & \cdots & 0& 0\\
0 & 1  & 0 & \cdots & 0 & 0\\
\vdots & \ddots & \ddots & \ddots & \ddots &\vdots\\ 
0 & 0 & 0 & \cdots & 1 &0 \\
\end{bmatrix}
\end{equation}

These operators are invertible, and the new stabilizer operators, $V_{d,\gamma}, P_{d,\gamma}$, will be defined below as particular tensor products of $X_d, Z_d$, and their inverses. For any single qubit error $\epsilon$, for it to change the ground state, it must commute with both $P$-type and $V$-type operators and thus commute with both $X_d$ and $Z_{d}$. However $Z_d^{-1}AZ_d=A$ if and only if $A$ is diagonal. But for a diagonal matrix, $D$, define $D':=X_d^{-1} D X_d$. As $D'_{i+1,i+1}=D_{i,i}$, $D'$ is a shifted form of $D$, which implies that all diagonal entries must be identical. Thus, $\epsilon$ is a multiple of the identity, and these stabilizers detect all non-trivial single-site qudit errors.

A non-trivial adjustment to the corresponding Hamiltonian is necessary so that the $1$-eigenspace of these complex operators corresponds to the ground state. This adjusted Hamiltonian is:
\begin{equation}H_{toric_{d}}(v)=\sum_{\gamma\in T_{k+1}} (2I-P_{d,\gamma}-P_{d,\gamma}^\dagger)(v) +\sum_{\alpha\in T_{k-1}}(2I-V_{d,\alpha}-V_{d,\alpha}^\dagger)(v)\end{equation} 

 Note that for a unitary operator $A$, $(2-A+A^\dagger)$ acts on eigenstate of $A$, $e_\lambda$, as $||1-\lambda||^2$. We can see this as, $||\lambda||=1$ and:
\begin{equation}
    ||1-\lambda||^2=(1-\lambda)(1-\bar\lambda)=1-\lambda-\bar\lambda-\lambda\bar\lambda=2-\lambda-\bar\lambda
\end{equation}
It is then clear that, as the $P$-type and $V$-type operators are unitary, the associated Hamiltonian has real non-negative eigenvalues.

Before defining $V_{d,\gamma}$ and $ P_{d,\gamma}$, we must provide an orientation on each $(k-1)$-cell, $k$-cell, and $(k+1)$-cell. For an $l$-cell, $p$, and $(l-1)$-cell, $q$, on $p$'s boundary, let $O(p,q)$ be $1$ if the induced orientation from $p$ on $q$ given by the boundary map agrees with the orientation on $q$. Otherwise let $O(p,q)=-1$. For each $(k-1)$-cell, $\alpha$, and $k$-cell, $\beta$, in the coboundary of $\alpha$, the operator $V_{d,\alpha}$ operates on $\beta$ by application of $Z_d^{O(\alpha,\beta)}$. Similarly, for a $(k+1)$-cell, $\gamma$, the operator $P_{d,\gamma}$ acts by $X_d^{O(\gamma,\beta)}$ operator on each $k$-cell, $\beta$, on the boundary of $\gamma$. Explicitly written out, these operators are:

\begin{equation}V_{d,\alpha}=\prod_{i\in \{\partial^\top \alpha\}} Z_{d,i}^{O(i,\alpha)}, \text{ and }P_{d,\gamma}=\prod_{i\in \{\partial^\top \gamma\}} X_{d,i}^{O(\gamma,i)}   \end{equation}

As for any stabilizer code, it is necessary to show that these operators commute. For any $(k-1)$-cell, $\alpha$, and $(k+1)$-cell, $\gamma$, $|\{\partial^\top \gamma\} \cap \{\partial \alpha\}|$ is always either zero or two $k$-cells. When it is zero, these trivially commute, as they act on different qudits. Otherwise, when there is an intersection of two $k$-cells, there are $16$ cases for all combinations of possible orientations of  $\alpha$,$\gamma$, and relevant $k$-cells $\beta_1$ and $\beta_2$. As reversing the orientations of $\alpha$ or $\gamma$ changes these operators to their inverse, and any operator commutes with an invertible operator, $A$, if and only if it commutes with $A^{-1}$, then without loss of generality, we can fix the orientations of $\alpha$ and $\gamma$.

We can then consider only the cases where $O(\beta_1,\alpha)=1$ and $O(\gamma,\beta_1)=1$. As, $\partial \partial \gamma=0$:

\begin{equation}O(\beta_1,\alpha)O(\gamma,\beta_1)+O(\beta_2,\alpha)O(\gamma,\beta_2)=0\end{equation}
Which implies in these cases that:
\begin{equation}O(\gamma,\beta_2)= -O(\beta_2,\alpha)\end{equation} 
So, the commutator is given by:

\begin{equation}\begin{aligned}
    [V_{d,\alpha},P_{d,\gamma}] & = V_{d,\alpha}^{-1}P_{d,\gamma}^{-1}V_{d,\alpha}P_{d,\gamma}\\
      & = [(Z_{d,\beta_1},X_{d,\beta_1}] [(Z_{d,\beta_2})^{O(\alpha,\beta_2)},(X_{d,\beta_2})^{-O(\alpha,\beta_2)}]
  \end{aligned}
\end{equation}
We can verify that the commutators of $Z_d,Z_d^{-1}, X_d, X_d^{-1}$ are: 
\begin{equation}[Z_d,X_d]=[Z_d^{-1},X_d^{-1}]=\zeta_{d}I\end{equation} 
\begin{equation}[Z_d,X_d^{-1}]=[Z_d^{-1},X_d]=\zeta_d^{-1}I\end{equation}

Therefore, the commutator on the $\beta_1$ component is $\zeta_d I$. As $O(\gamma,\beta_2)= -O(\beta_2,\alpha)$, the commutator restricted to the qudit on $\beta_2$ is either $[Z_d,X_d^{-1}]$ or $[Z_d^{-1},X_d]$ which regardless is equal to $\zeta_d^{-1}I$. Restricted to all other qudits, the commutator acts as the identity. So, we have shown that the commutator of these operators over the whole Hilbert space is $\zeta_d \otimes \zeta_d^{-1}\otimes I=I$, and so, these operators commute.

After defining the Hilbert space and the Hamiltonian so that the system of stabilizers is well defined, we must show that this system encodes $\C[H_k(T,\Z_d)]$. The total Hilbert space is isomorphic to the group algebra of $k$-chains over $\Z_d$ using the same isomorphism as the $\Z_2$ case. Each $P_{\gamma,d}$ operator acts on a basis element $\ket{v}$, corresponding to $k$-chain $\tilde{v}$, by taking it to $\ket{v+\partial \gamma}$. As before, this enforces that the coefficients of homologous chains are equivalent, and for any chain $\theta$, the equal superposition of all chains homologous to $\theta$ is acted on as the identity. Similarly, the $V_{\alpha,d}$ operators acts as the identity on a basis vector $\ket{v}$, corresponding to $k$-chain $v$, if and only if the oriented sum of $k$-cells on the coboundary of each $(k-1)$-cell, $\tilde{v}_i$, is exactly $0$, i.e., only if $\sum_{i\in \{\partial^\top \alpha\}} O(i,\alpha)\tilde{v}_i=0$. Therefore, the collection of $V_{\alpha,d}$, and $P_{\gamma,d}$ operators acts on a vector $\ket{\psi}=\sum_{i=1}^l c_i\ket{r_i}$ as the identity if and only if each $r_i$, with non-zero coefficient $c_i$ is a cycle and the coefficient of all homologous $r_i$ are equivalent. As in the $\Z_2$ case, this implies that the ground state is isomorphic to $\C[H_k(T,\Z_d)]$.

After showing the result over $A=\Z_d$, the next abelian group to consider is $\Z$. To begin, place infinite-dimensional qudits on each $k$-cell. Note that this is somewhat nonphysical, as we cannot use the structure of Hilbert spaces or Hamiltonians here without accounting for infinite sums of computational basis vectors, which we will address more thoroughly in the following subsection. Define these infinite-dimensional vector spaces as maps, $f$, from $\Z$ to $\C$. To define a basis for this space, for an integer $n$, let $a*\ket{n}$ be the map that takes $n$ to $a$ and every other integer to $0$. Therefore, the vector space is identified with maps, $f$, from $k$-chains over $\Z$ to $\C$. This code will use a similar construction as the $\Z_d$ codes, except that we extend the Pauli matrices into infinite-dimensional operators. The new $X_\Z$ operator is the shift operator, such that:

\begin{equation}X_\Z(f)(n)=f(n+1)\end{equation}
In terms of the above basis, this map is:

\begin{equation}X_\Z(\ket{n})=\ket{n+1}\end{equation}
It is invertible, with its inverse given by 

\begin{equation}X_\Z^{-1}(\ket{n})=\ket{n-1}\end{equation}
Given some specified irrational $r\in \R$, the new $Z_{\Z,r}$ operator is the diagonal operator: 
\begin{equation}Z_{\Z,r}(f)(n)=(e^{2\pi i r})^n f(n)
\end{equation}
Which, in the above basis, is:
\begin{equation}
Z_{\Z,r}(\ket{n})=(e^{2\pi i r})^n\ket{n}\end{equation} 
Its inverse is given by:
\begin{equation}Z^{-1}_{\Z,r}(\ket{n})=(e^{-2\pi i r})^n\ket{n}\end{equation} 
Restricted to a finite-dimensional snapshot, these operators would act like the matrices: 

\begin{equation}Z_{\Z,r}=
\begin{bmatrix}
	\vdots & \ddots & \ddots & \ddots & \ddots \\ 
	\cdots & (e^{2\pi i r})^{-1} & 0 & 0 &  \cdots \\
	\cdots & 0 & 1 & 0 & \cdots \\ 
	\cdots & 0 & 0 & e^{2\pi i r} & \cdots \\
	\vdots & \ddots & \ddots & \ddots & \ddots\\ 
\end{bmatrix}
\end{equation},
\begin{equation}
X_\Z=
 \begin{bmatrix}
	\vdots & \ddots & \ddots & \ddots & \ddots & \ddots\\ 
	\cdots & 0 & 0 & 0 & 0 & \cdots  \\ 
	\cdots & 1 & 0 & 0 & 0 & \cdots \\
	\cdots & 0 & 1 & 0 & 0 & \cdots  \\
 	\cdots & 0 & 0 & 1 & 0 & \cdots  \\ 
	\vdots & \ddots & \ddots & \ddots & \ddots &\ddots\\ 
 \end{bmatrix}
\end{equation}

We construct the $P_\Z$ and $V_{\Z,r}$ operators as we did in the $\Z_d$ case, for which the same argument shows that these stabilizer operators commute. For instance, consider a $(k-1)$-cell, $\alpha$, with cells $j_1,\dots, j_m$ in the coboundary, such that, without loss of generality, their orientations agree with $\alpha$, i.e., $O(j_1,\alpha)=1, O(j_2,\alpha)=1, \dots, O(j_m,\alpha)=1$. Then $V_{\Z,r,\alpha}$ acts as the identity on map, $f:\Z[T_k]\rightarrow \C$, if and only if, $f$ is a function that is non-zero only when the variables $j_1\dots j_m$ have:
\begin{equation}\sum^m_{i=1} j_i=0\end{equation}

 Alternatively, consider a $(k+1)$-cell, $\gamma$, with $j_1,\dots, j_m$ in the boundary, such that, without loss of generality, their orientations agree on $\gamma$, i.e., $O(\gamma,j_1)=1, O(\gamma,j_2)=1, \dots, O(\gamma,j_m)=1$. Then $P_{\Z,\gamma}$ acts as the identity on maps $f:\Z[T_k]\rightarrow \C$, if and only if, 
\begin{equation}f(j_1,\dots j_m,\dots)=f(j_1+1,\dots j_m+1,\dots)\end{equation}

Then the $1$-eigenspace of both of these operators is the subspace of maps which only send $k$-cycles over $\Z$ to non-zero elements of $\C$ and send homologous cycles to the same elements of $\C$. Therefore, this code-space is isomorphic to the set of maps from the homology classes over $\Z$ to $\C$ and so encodes $\C[H_k(T,\Z)]$. 

To finish the construction, as all finitely generated abelian groups, $A$, are direct sums of the integers, $\Z$, or finite cyclic groups, $\Z_d$, this lemma is then equivalent to the statement that given an encoding of $H_k(T, A)$ and $H_k(T, B)$ we can encode $H_k(T, A\oplus B)$.  As the homology over a direct sum is equivalent to the direct sum of the homologies, i.e., \begin{equation}H_k(T, A\oplus B)=H_k(T, A)\oplus H_k(T, B)\end{equation}
It must be the case that: 
\begin{equation}H_k(T,\Z^{r_0}\oplus^m_{i=1}\Z^{r_i}_{d_i})=\bigoplus^{r_0}_{s=1}H_k(T,\Z)\bigoplus^m_{i=1}\bigoplus^{r_i}_{s=1} H_k(T,\Z_{d_{i}})\end{equation}
As $\C[P\oplus Q]=\C[P]\otimes \C[Q]$, to encode homology over any finitely generated abelian group, we must place all of the individual QECCs on the same triangulation. The stabilizers are then the collection of all stabilizers in all codes, and each stabilizer acts only on the qudits of its respective code. If the group is finite, the Hamiltonian is the sum of each given Hamiltonian.
\end{proof}

\subsection{Making the Integer Case Physical}\label{annoying:integer}

In the QECC described above, the codes over $\Z_d$ correspond to a gapped Hamiltonian schema, i.e., the difference in energy between the ground state and the first excited states does not converge to zero and remains finite regardless of the choice of a given triangulation of a topological space or the number of errors. However, in the encoding of homology over the integers, we did not define a Hilbert space or Hamiltonian. This was not defined because the magnitude of the encoded states would be unbounded, as they are infinite sums of computational basis vectors with equal coefficients. To define the Hilbert space of square summable vectors in $\C[\Z]$, $H_\beta$, placed on each $k$-cell, $\beta$, we define the inner product:

\begin{equation}\langle f,g\rangle=f(0)\overline {g(0)}+\sum_{i=1}^{\infty}(f(i)\overline{g(i)}+f(-i)\overline{g(-i)})\end{equation}

The induced inner product for finite tensor products defines the inner product structure on the entire Hilbert space, $H$, of square summable vectors. If we want to use this Hilbert space structure, we must alter the above construction so that the stabilizer subspace of the $P$-type and $V$-type operators contains only bounded vectors. This will allow us to construct a Hamiltonian whose ground state is isomorphic to $\C[H_k(T,\Z)]$ and whose vectors have a bounded norm and are thus physically well-defined. The basis vectors $\ket{n}$ as described above now define a Schauder basis for this Hilbert space, allowing us to use the same $V_{\Z,r}$-type operator as in the above construction. However, we will change the $P_{\Z}$-type operators by adjusting the generalized Pauli $X_\Z$ operators. Define,

\begin{equation}\tilde{X}_\Z(\ket{n})=
\begin{cases}
    \frac{|n|+1}{|n|} \ket{n+1}, n<0\\
    \frac{n+1}{n+2} \ket{n+1}, n\geq 0
\end{cases}
\end{equation}

The eigenvectors of this operator are the functions:

\begin{equation}f_{a}(i)=\frac{a^{i}}{|i|+1}\end{equation}

The magnitude of $f_{a}$ is given by:

\begin{equation}\sqrt{1+\sum^\infty_{i=1} \frac{|a|^{2i}+|a|^{-2i}}{n^{2i}}}\end{equation}

We can then see that for this to remain physical, i.e., for the magnitude of $f_{a}$ to remain bounded, we must have $||a|| = 1$. Unfortunately, this operator is no longer unitary; however, on this set of eigenvectors, we still have that $2-\tilde{P}+\tilde{P}^{-1}$ acts on such an eigenstate of $\tilde{P}$, $e_\lambda$, with as $||1-\lambda||^2$. We can see this as $\tilde{P}^{-1} e_\lambda=\lambda^{-1}e_\lambda$ for all eigenvectors $e_\lambda$. As $||\lambda||=1$ in the square summable eigenvectors, $\bar\lambda=\lambda^{-1}$. It is then clear that the associated Hamiltonian has real non-negative eigenvalues restricted to square summable vectors.

\subsubsection*{The Action of $\tilde{P}_\Z$}
The $\tilde{X}_\Z$ operator is invertible just as the previous $X_\Z$ was. Using the same argument as in the finite-dimensional case, the $\tilde{P}_{\Z,\gamma}$ operator commutes with $V_{\Z,r}$ as well. To better understand the action of $\tilde{P}_{\Z,\gamma}$ on the whole Hilbert space, consider $(k+1)$-cell, $\gamma$ with $\partial \gamma= j_1\dots j_m$ oriented so as to agree with the orientation of $\gamma$. For a function $f$ in the Hilbert space, we define $f(r_{j_1},\dots, r_{j_m},\dots, r_{j_n})$ as the result of $f$ on the chain $r_{j_1} j_1+r_{j_2}j_2+\dots r_{j_m}j_m+\dots r_{j_n}j_n$. If without loss of generality the $r_{j_{1}}\dots r_{j_{l}}$ are non-negative then:

\begin{equation}\tilde{P}_{\Z,\gamma}(f)(r_1,\dots r_m,\dots,r_n)=\left( \prod_{i=1}^l \frac{r_i+1}{r_i+2} \right)\left(\prod_{i=l+1}^m \frac{|r_i|+1}{|r_i|}\right)f(r_1+1,\dots r_m+1,r_{m+1}\dots,r_n)\end{equation}

Therefore, we determine that for each $k$-chain, $\beta$, the 1-eigenspace of the $\tilde{P}_{\Z,r}$ operators is given by the following weighted sum over cycles homologous to $\beta$:

\begin{equation}h_{[\beta]}=\sum_{\beta_i\in [\beta]}\ket{\beta_i}\left(\prod_{j\in T_k} \frac{1}{|(\beta_i)_j|+1}\right) \end{equation}

Where $(\beta_i)_j$ is the weight in $\beta_i$ of $k$-cell $j$. To see this note, that $\tilde{P}_{\Z,\gamma}$ takes each element in this sum to a unique other element in this sum, and each vector with this property can be written as a sum of $h_{[v]}$s for some vectors $v$. Therefore, the intersection of the $1$-eigenspace of the $\tilde{P}_{\Z}$ and $V_{\Z.r}$ operators are the $h_{[\beta]}$ where $\beta$ is a $k$-cycle. Let $n=|T_k|$, be the number of $k$-cells contained in $T$. Then by considering the sum of the weighted superposition of each $k$-cycle, the magnitude of each $h_{[\beta]}$ is bounded above by:

\begin{equation}
\begin{split}
    \norm{h_{[\beta]}}^2&\leq\sum_{i_1\in \Z} \dots \sum_{i_n\in \Z}  \left(\prod_{i=1}^n \frac{1}{|i_j|+1}\right)^2\\
    &\leq 2\sum_{i_1=0}^\infty \dots 2\sum_{i_n=0}^\infty \left(\prod_{i=1}^n \frac{1}{|i_j|+1}\right)^2\\
     &\leq 2^n\left(\sum_{i_1=0}^\infty \dots\sum_{i_{n-1}=0}^\infty \left(\prod_{i=1}^{n-1} \frac{1}{|i_j|+1}\right)^2\left(\sum_{i_n=0}^\infty \left( \frac{1}{|i_j|+1}\right)^2 \right)\right)\\
     &\leq 2^n \left(\frac{\pi^2}{6}\right)\sum_{i_1=0}^\infty \dots \sum_{i_{n-1}=0}^\infty\left(\prod_{i=1}^{n-1} \frac{1}{|i_j|+1}\right)^2\\
     &\leq 2^n \left(\frac{\pi^2}{6}\right)^n\\
\end{split}
\end{equation}

This new construction, if used to construct a QECC, would have no gap between the ground state and the higher energy states, as the values of $||(e^{r\pi ni})-1||$ can become arbitrarily small as $n$ varies over the integers, which implies that the code will never reach the ground state.

\subsection{A Third Choice of $P$-type and $V$-type Operators}
We can construct simpler operators than those above by using operators that are not tensor products of single-qudit operators. In this new construction, there is an identical setup of Hilbert space and orientations on the relevant cells. For finite $d$, the $V_{\alpha,d}$ are replaced with the operator, $\bar{V}_{\alpha,d}$, such that for any $k$-chain, $v$, where $\sum_{i\in \partial^\top \alpha} O(i,\alpha)v_i=0$, $\bar{V}_{\alpha,d}$ acts as the identity. On all other basis vectors, $\bar{V}_{\alpha,d}$ acts as the $0$ operator. Similarly, for any $k$-chain, $v$, $\bar{P}_{\gamma,d}$ act as the identity on all vectors $h(v)=\sum^d_{i=0} \ket{v+i\partial \gamma}$, and as the zero operator on the orthogonal complement to the subspace generated by the $h(v)$ vectors. For the integer case, we can instead have $\bar{P}_{\gamma,\infty}$ act as the identity on vectors $h(v)=\sum_{i=\Z} \left(\frac{1}{2}\right)^{\sum_{j\in T_k} |(v+i \partial \gamma)_j|}\ket{v+i\partial \gamma}$ and as the zero operator on the orthogonal complement to the subspace generated by the $h(v)$ vectors.

These operators commute because if a chain is a cycle, so is any homologous chain. The above argument also shows that these stabilizers encode $H_k(T, A)$ for any finitely generated abelian group $A$. Moreover, the Hamiltonian, in this case, is exactly: 
\begin{equation}\bar{H}_{toric_d}= \sum_{i\in T_{k+1}} (1-\bar{P}_i)+\sum_{i\in T_{k-1}} (1-\bar{V}_i)\end{equation}

The main disadvantage of this formulation is that these operators are non-decomposable as single qudit operators on all qudits on the boundary or coboundary of a cell, which can be arbitrarily large for a general cellulation.

\subsection{Encoding Cohomology}\label{cohomsect}

The model for encoding cohomology will be nearly identical to the above construction except instead of $P$-type operators using variant Pauli-$X$ operators, and $V$-type operators using variant Pauli-$Z$ operators, $P_\gamma$-type operators will act using the variant $Z$ operators on the boundary of a $(k+1)$-cell, $\gamma$, and $V_\alpha$-type operators will act using variant $X$ operators on the coboundary of a $(k-1)$-cell, $\alpha$. To prove that this encodes the cohomology, $H^k(T, A)$, identify the Hilbert space used in the homology case with the group algebra over $\C$ of cochains over $A$. This isomorphism $i$ is defined by:

\begin{equation}i(\ket{\sum_{i\in T_k} a_i i})=\ket{\sum_{i\in T_k} a_i e^i}\end{equation}

Above, $\ket{\sum_i a_i i}$ is the basis vector over $\C$ corresponding to the chain 
$\sum_{i\in T_k} a_i i$, and $\ket{\sum_{i\in T_k} a_i e^i}$ is the basis vector over $\C$ corresponding to the cochain $\sum_{i\in T_k} a_i e^i$, such that $\left(\sum_{i\in T_k} a_i e^i\right)(j)=a_j$.

For a finite number of $k$-cells, this is an isomorphism as all chains are a linear combination of $i$, and finite cochains are a linear combination of $e^i$. As $d(a_i e^i)=\sum_{j \in \partial^\top i} a_i^{o(j, i)}e^{j}$, the $P_\gamma$ operators acting on $(k+1)$-cells ensure that all encoded basis vectors must be cocycles. Similarly, the $V_\alpha$ operators acting on $(k-1)$-cells ensure that all cohomologous basis vectors have equivalent coefficients.

A similar way to define the cohomology is by considering the dual of a triangulation $\widehat{T}$ of a manifold, as defined by taking each $(n-k)$-cell, $\Delta \in T$, to a dual k-cell, $\widehat{\Delta}$ \cite{dualcellcomplexpaper}, where: 

\begin{equation}\hat{\partial}(\widehat{\Delta})= \sum_{i\in \partial^\top \Delta} \widehat{O}(\Delta,i)\widehat{i},\text{ and }\widehat{O}(\Delta,i)=O(i,\Delta)\end{equation} 

This dual will not take simplicial complexes to CW-complexes, as for a boundary $l$-cell, $i\in T$, $\hat{\partial} \widehat{i}$ has the topology of the half-open $(n-l-1)$-disk. To address this, we define the closure of $\widehat{T}$, referred to as $cl(\widehat{T})$, as $\widehat{T}\cup \partial \widehat{T}$, where $\partial \widehat{T}$ is defined as the dual of $\partial T$. For each, $i\in \partial T$, $\widehat{i}\in \partial \widehat{T}$ is pasted to the boundary of $\widehat{i}\in \widehat{T}$, with orientations inherited from $\widehat{T}$. This closure forces the construction to become a well-defined $CW$-complex at the cost of changing the cohomology. 

\begin{figure}[h]
    \centering
    \includegraphics[width=.7\linewidth]{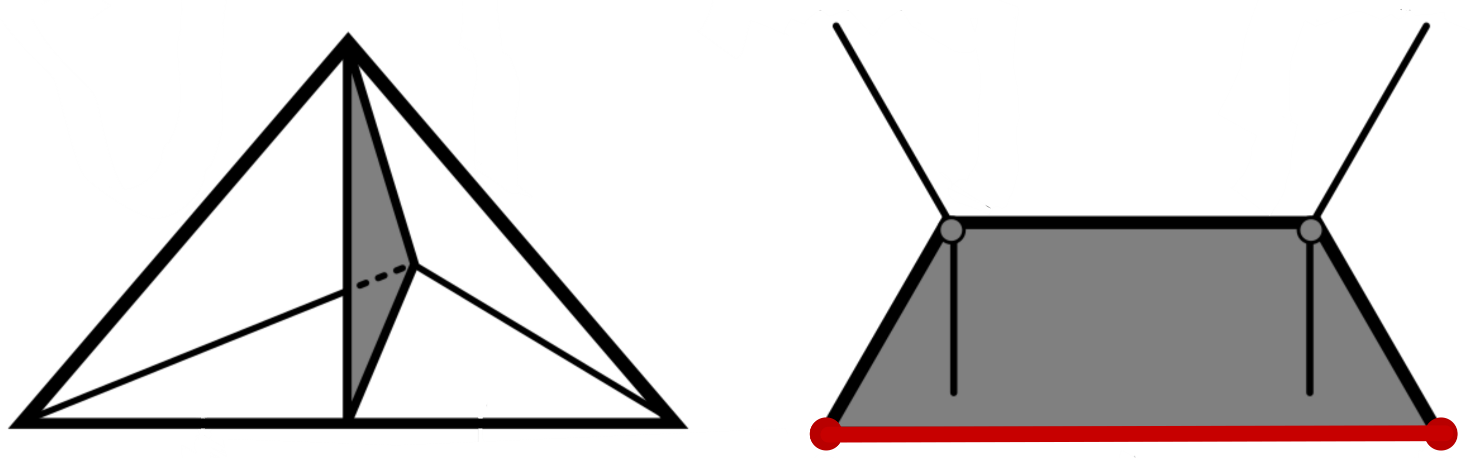}
    \caption{The left side is a snapshot of the triangulation of two $3$-cells on the boundary of a $3$-dimensional triangulation (with the boundary facing the viewer), and the right side is the dual of this snapshot. The red cells on the right are in the closed dual triangulation but not the open dual triangulation. Each $3$-cell on the left goes to a corresponding shaded vertex on the right. Similarly, the shaded face on the left goes to the corresponding shaded face on the right. On the left, the two $2$-cells on the boundary go to the corresponding edges from the shaded vertex to the red vertex on the right. }
    \label{fig:dualclosedcell}
\end{figure}

The (open) dual triangulation is well-defined as a formal chain complex even though it is not well-defined as a $CW$-complex outside of compact closed manifolds without boundary. Assuming a finite triangulation, we have that $H_k(\widehat
T)\cong H^{n-k}(T)$ \cite{dualcellcomplexpaper}.

\section{The Structure of the Errors}

\subsection*{Types of Errors}
In the $H_k(T,\Z_d)$ codes, there are two types of stabilizer operators, and so we will be able to decompose each error as a sum of two types of errors, which we call $P$-type/$Z$-errors and $V$-type/$X$-errors, respectively. The following discussion is valid in the cohomological case by switching this naming convention, i.e., we have $P$-type/$X$-errors, and $V$-type/$Z$-errors.

In an encoding of $H_k(T,\Z_d)$, $P$-type errors occur because of an application of a $Z_{d}$ operator to a single qudit, which breaks the $P$-stabilizers but not the $V$-stabilizers, which is why we also referred to them as $Z$-errors. For instance, applying $Z_d^i$ to a qudit would cause a $Z/P$-error on the cobounding $(k+1)$-cells. As the orientation of the qudit relative to nearby cells determines how the error affects those cells, we can assign each error a defined degree, $d$, and orientation, $U$. For example, a $P$-type error on a $k$-cell, $\beta$, of degree-$(l)$, with orientation $U|_\beta$, is defined as when the operator $X_d^{O(U|_{\beta},\beta)l}$ is applied to the qudit at $\beta$, at some arbitrary ground state. 

An error caused by the operator $X_d^{O(U|_{\beta},\beta)l}$ is equivalent to $l$ errors caused by the same operator $\prod_{i=1}^l X_d^{O(U|_\beta,\beta)}$, so if we can correct an error on $\beta$, with degree-$(1)$, and orientation $U$, we can correct an error of the same type on $\beta$ with degree-$(l)$ and orientation $U$, by repeating the process $l$ times. For the sake of manipulating errors and simplifying notation, we take the perspective that a degree-$(l)$ error is the superposition of $l$ degree-$(1)$ errors that happen to be on the same cell with the same orientation.

To ease notation, we will define $O(\gamma,U|_\beta)=O(\gamma,\beta)O(U|_\beta,\beta)$. if $X_d^{O(U|_\beta,\beta)}$ is applied to $\beta$, applying $P_{d,\gamma}^{O(\gamma,U|_\beta)}$ operator to an adjacent $(k+1)$-cell, $\gamma$, can fix the resultant error on $\beta$. However, this results in applying the error $X_d^{O(\gamma,U|_\beta)O(\gamma,i)}$
on all other $k$-cells, $i$, on the boundary of $\gamma$. If in this state, no adjacent $P$ stabilizers are violated, then this action acts as the identity. In this way, a detected $X$-error only indicates the boundary of the error cells, not the cells themselves. 

We will switch between the triangulation $T$ and its dual $\widehat{T}$ when considering $V$ and $P$-type errors, because of the similarity of the structure of their errors. For instance, when applying a $V_{d,\alpha}$ operator on a $P$-error on a $k$-cell, $\beta$. Using the dual picture, the error on $\widehat{\beta}\in \widehat{T}_{n-k}$ is fixed at the potential cost of applying an error to all elements on $\hat{\partial}\widehat{\alpha}\setminus\widehat{\beta}$. If all adjacent stabilizers are satisfied, then this operator acts as the identity. Therefore, just as in the non-dual case, in the dual cellulation, the detection of $V$-errors only reveals the boundaries of the error dual-$(n-k)$-cells, not the original error cells. In the non-dual picture, these are the coboundaries of the original $k$-cells of the error. 

\begin{figure}[ht]
    \centering
    \includegraphics[width=.7\linewidth]{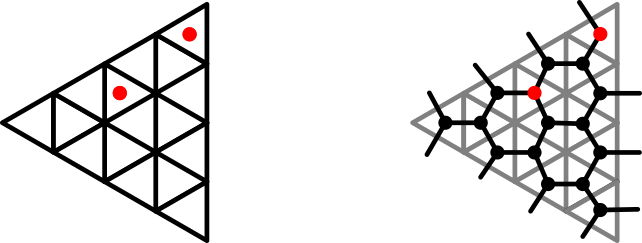}
    \caption{A visual representation of two $P$-errors in the normal triangulation on the left and in the dual triangulation on the right.}
    \label{fig:correctingerrorsindual}
\end{figure}

\subsection{Combining Errors and Error boundaries}
 In the toric code, stabilizers that are unsatisfied act similarly to point particles. For instance, if a $0$-cell operator is not satisfied, then applying a correction by adjusting a neighboring qubit satisfies this operator at the cost of affecting the adjacent stabilizer. If this nearby stabilizer were satisfied, then this action would break it. This fix and translation of the broken stabilizer can be thought of as the \textit{particle} moving. Similarly, a broken $2$-cell stabilizer, if fixed, moves to an adjacent $2$-cell sharing the adjusted qubit. If the neighboring cell's stabilizer is broken, then this fix annihilates the particles. More generally, in any dimension, an error on a qudit creates broken $V$-stabilizers on the $(k-1)$-cells on the boundary and $P$-stabilizers on the $(k+1)$-cells on the coboundary. 

We define two adjacent $V$-errors of the same degree, $u_1|_{\beta_1}$ on $\beta_1$ and $u_2|_{\beta_2}$ on $\beta_2$, to share an orientation if the $(k-1)$-cell, $\alpha$, in their intersection disagrees with one of the orientations, and agrees with the other. For instance, if these errors were the only $V$-errors on these cells, then $\alpha$'s stabilizer would be satisfied, as, when applied to the ground state, the stabilizer:

\begin{equation}Z_{d,\beta_2}^{O(\beta_1,\alpha)} Z_{d,\beta_2}^{O(\beta_2,\alpha)} \prod_{i\in \{\partial^\top \alpha\}\setminus \{\beta_1,\beta_2\}} Z_{d,i}^{O(i,\alpha)}\end{equation}
commutes with the error:

\begin{equation}X_{d,\beta_1}^{O(u_1|_{\beta_1},\beta_1)} X_{d,\beta_2}^{O(u_2|_{\beta_2},\beta_2)}\end{equation}

These operators commute as, $O(u_1|_{\beta_1},\beta_1)O(\beta_1,\alpha)=-O(u_2|_{\beta_2},\beta_2)O(\beta_2,\alpha)$. Implying the commutator is $\zeta_d^{\pm 1}\otimes \zeta_d^{\mp 1}\otimes I=I$. As $V_\alpha$ acts as the identity on the given cell, the stabilizer is satisfied. Conversely, two adjacent $P$-errors of the same degree, $u_1|_{\beta_1}$ on $\beta_1$ and $u_2|_{\beta_2}$ on $\beta_2$, share an orientation if the $(k-1)$-cell in their intersection agrees with both of $u_1|_{\beta_1},u_2|_{\beta_2}$. If these were the only $P$-errors on these cells, then a $(k+1)$-cell's stabilizer containing these cells in its boundary would also be satisfied. As $\partial \partial \gamma=0$, for every $(k-1)$-cell, $\alpha$, such that $ \{\beta_1,\beta_2\}=\{\partial \gamma\} \cap\{\partial^\top \alpha\}$, $O(\gamma,\beta_1)O(\beta_1,\alpha)=-O(\gamma,\beta_2)O(\beta_2,\alpha)$. Therefore as, 

\begin{equation}
    O(u_1|_{\beta_1},\beta_1)O(\beta_1,\alpha)=O(u_2|_{\beta_2},\beta_2)O(\beta_2,\alpha)
\end{equation}
We have that: 
\begin{equation}O(\gamma,\beta_1)/O(\gamma,\beta_2)=-O(\beta_2,\alpha)/O(\beta_1,\alpha)=-O(u_1|_{\beta_1},\beta_1)/O(u_2|_{\beta_2},\beta_2)\end{equation}
Since these values are all $\pm 1$, \begin{equation}
    O(\gamma,\beta_1)O(u_1|_{\beta_1},\beta_1)=-O(u_2|_{\beta_2},\beta_2) O(\gamma,\beta_2)
\end{equation}
So, just as before the commutator is $\zeta_d^{\pm 1}\otimes \zeta_d^{\mp 1}\otimes I=I$ and this stabilizer commutes with the error and therefore acts as the identity on the ground states, this shows that even though an error operator was applied to a cell on the boundary of $\gamma$, the error itself is undetectable on $\gamma$.

We can again see similarities with the $V$-error case for $\widehat{T}$, since in any simplicial complex for any two $k$-cells, $\beta_1$ and $\beta_2$, $|\{\partial^\top \beta_1\} \cap \{\partial^\top \beta_2\}|\leq 1$, in $\widehat{T}$ between any two $(n-k)$-cells there is at most one unique $(n-(k+1))$-cell. So, when two $P$-type errors share an orientation, we have:
\begin{equation}O(\gamma,\beta_1)O(u_1|_{\beta_1},\beta_1)=-O(u_2|_{\beta_2},\beta_2)O(\gamma,\beta_2)\end{equation}

In $\widehat{T}$, we can define this cell as the $(n-(k+1))$-cell in their common boundary that agrees with one error and disagrees with the other. Given a definition of common orientation on two separate errors, we can now stitch these errors together into a singular structure, given by a connected, oriented pseudo-manifold.

\begin{defi}[Oriented $m$-pseudo-manifold]
A $m$-dimensional connected pseudo-manifold (with boundary) $X$ is a topological space with a triangulation such that:
\begin{enumerate}
    \item $X$ as a set is the union of all of its $m$-simplices.
    \item Every $(m-1)$-simplex is the face of exactly two or one $m$-simplex, with one $m$-simplex, the case for boundary $m$-simplices.
    \item $X$ is connected in that any two $m$-cells, $\beta_1$, $\beta_2$, are connected by a chain of $m$-cells, $\{u_1=\beta_1,u_2,\dots u_l=\beta_2\}$ such that for any $i$, $u_i$, and $u_{i+1}$ share a common $(m-1)$-cell in their boundary.
\end{enumerate}
For the remainder of the discussion, each mention of a pseudo-manifold will refer to a pseudo-manifold with (possibly empty) boundary, unless mentioned otherwise.

An oriented $m$-pseudo-manifold has a choice of orientation on all $(m-1)$ and $m$-cells such that any two $m$-cells that share a $(m-1)$-cell have opposite orientations on said cell \cite{spanier1989algebraic}.
Similarly, define a connected, oriented $m$-cellulation to be a CW-complex with the above structure.
\end{defi}

Using this structure, we can define the following method for stitching together error cells:

\begin{defi}[$V$-error sub-pseudo-manifold]
A (degree-$(1)$) $V$-error sub-pseudo-manifold is a continuous simplicial immersion, $f$ from an oriented $k$-pseudo-manifold $R$ to $T$, such that each oriented $k$-cell goes to an oriented $k$-cell. The boundary of the error is $\partial f(R)$. The error operator associated with this is:
\begin{equation}\prod_{\rho \in R_k} X_{d,f(\rho)}^{O(f(\rho),\rho)}\end{equation}
\end{defi}

We define a (degree-$(1)$) $P$-error sub-pseudo-manifold similarly on the dual triangulation:
\begin{defi}[$P$-error sub-pseudo-manifold on the dual triangulation]
    On the dual triangulation, $\widehat{T}$, a (degree-$(1)$) P-error sub-cellulation, corresponds to a continuous immersion, $f$ from a connected oriented $(n-k)$-cellulation, $R$, to $cl(\widehat{T})$ such that each $(n-k)$-cell goes to an $(n-k)$-cell in $\widehat{T}$. The boundary of the error is $\partial f(R)$. Note that no $(n-k)$-cell is sent to an $(n-k)$-cell in $\partial\widehat{T}$; however, the closure is necessary to ensure that both sides of the map are CW complexes. 
    This results in the error operator: \begin{equation}
        \prod_{\rho \in R_{n-k}} Z_{d,\widehat{f(\rho)}}^{O(f(\rho),\rho)}
    \end{equation}
\end{defi}

We pass to the closed dual so that the resulting structure is a valid CW-complex \cite{dualcellcomplexpaper}. These structures lead us to the following lemma for decomposing and stitching together collections of errors. 

\begin{lem}
Given a (degree-$(1)$) $V$-error sub-pseudo-manifold, $R$, take the chain associated with its image, $\rho$; if the error associated with $R$ is the only error on the manifold, then the only broken stabilizers are the $V$-stabilizers on the $(k-1)$-cells on the boundary of $\rho$. Moreover, this error is detectable only up to the boundary of the non-broken $(k+1)$-cells' stabilizers on $T$. 
Similarly, given a degree-$(1)$ $P$-error sub-pseudo-manifold, $R$, take the chain associated to the dual of its image, $\widehat{\rho}\in cl(\widehat{T})$. If the error associated with $R$ is the only error on $T$, then the only broken stabilizers are the $P$-stabilizers on the $(n-(k+1))$-cells in $\widehat{T}$ on the boundary of $\widehat{\rho}$. Moreover, this error is detectable only up to the coboundary of the non-broken $(n-(k-1))$-cells' stabilizers on $\widehat{T}$. Also, given a set of $V$ and $P$ errors and compatible orientations on chosen cells, there is a set of maximally connected (degree-$(1)$) $V$ and $P$ error sub-pseudo-manifolds with the given orientations, such that the boundary of the error corresponds to the broken stabilizer cells.
\end{lem}

\begin{proof}
We showed above that, given a $V$-error tracing a manifold, the broken stabilizers associated with it are its boundary, and similarly for a $P$-error tracing a manifold in the dual triangulation. So, all that remains to be shown is that for any set of $V$ and $P$ errors and chosen orientations of compatible cells, there is a set of maximally connected (degree-$(1)$) $V$ and $P$ error sub-pseudo-manifolds whose boundaries are the broken stabilizers.

Let $U=\{u_1|_{\beta_1},\dots u_l|_{\beta_l}\}$ be a set of $V$ and $P$ errors on $T$. Where $u_i|_{\beta_i}$ is some product of $X_{d,\beta_i}$, and $Z_{d,\beta_i}$. We can take a degree-$(1)$ $V$-error, $U_1$ on $\beta_i$, oriented accordingly, and if there is a similar type error on an adjacent cell, increase the size, or number of corrupted cells, of $U_1$ by including the other cell, reversing orientations if necessary. We are free to continue this process iteratively, adding error cells to the boundary that match the orientation of $U_1$ until $U_1$ is of maximal size. We can then subtract these errors from the total set $U$. If there are still $V$-errors in $U$, we repeat the process by defining $U_2$ as a separate error using orientations that agree with those in $U_1$. We continue to repeat until there are no such errors in these cells. Then, repeat this process until there are no errors left unaddressed in any cell. Finally, we do the same for $P$-errors in the dual triangulation $\widehat{T}$.

This construction defines a set of $V$-error sub-pseudo-manifolds and $P$-error sub-cellulations, such that the broken stabilizers on the manifold are precisely the boundary of these maximal sets. To see this, consider a cell whose stabilizer is not broken on such a boundary; then, it must be adjacent to at least two different maximal error cellulations. If any of these correspond to the same orientation, then we can connect them, forming a larger maximal cellulation, which contradicts our initial assumption. Suppose the orientations disagree over $\Z_d$. In that case, we can define the error as $(d-1)$ degree-$(1)$ errors of the opposite orientation, and then paste this error into the error pseudo-manifold, making it larger, which contradicts our assumptions. If the orientations were to disagree over $\Z$, then we would define the error as a degree-$(-1)$ error of the opposite orientation, and if there are no broken stabilizers on the boundary, the adjacent error is also a degree-$(-1)$ error.

\begin{figure}
    \centering
    \includegraphics[width=\linewidth]{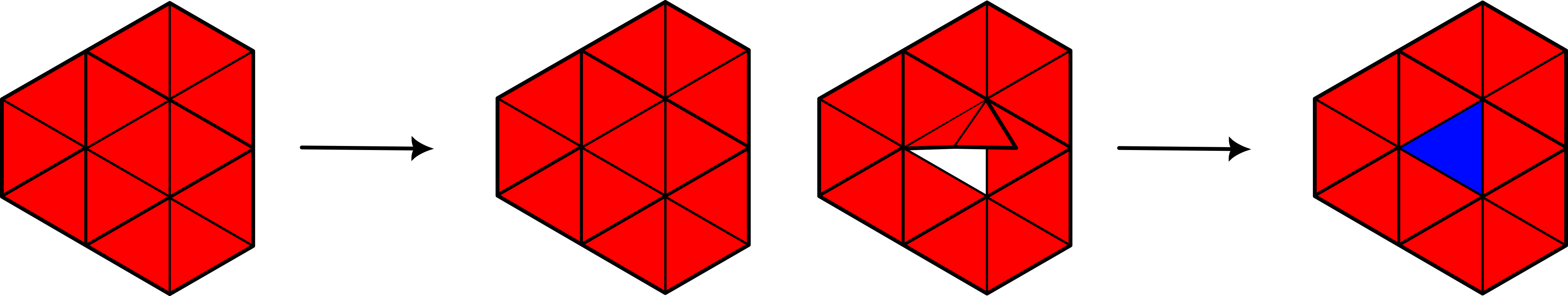}
    \caption{This is a visual representation of error sub-pseudo manifolds over $\Z_3$, where red indicates a degree-$(1)$ error and blue indicates a degree-$(2)$ error}
    \label{fig:psuedomanif}
\end{figure}

On the other hand, consider a cell whose stabilizer is broken yet not on the boundary of such a sub-cellulation. Therefore, it is not in the sub-cellulation at all, or it is on the interior. By construction, any cell that is only in the interior of all error sub-manifolds does not have a broken stabilizer, and any cell not in the interior or the boundary is not adjacent to any errors and, therefore, cannot have a broken stabilizer as the surrounding data is unchanged from how it was in the ground state. 
\end{proof}

\begin{defi}[Error Boundary/Erb]
    We will use the term $V$-erb or $V$-type error boundary to refer to the boundary of some $V$-error sub-pseudo-manifold. Similarly, we define a $P$-erb or $P$-type error boundary as the boundary of a $P$-error sub-cellulation in $\widehat{T}$. 
\end{defi}

Just as we can identify degree-$(l)$ errors with orientation $U$, with $l$ degree-$(1)$ errors of the same orientation, we can identify a degree-$(l)$ erb of orientation $U$ as $l$ degree-$(1)$ erbs of orientation $U$ that happen to be in the same location. For example, consider three $V$-erbs in the toric code of degrees of $-p,-q$, and  $p+q$.
We could interpret the $p+q$ erb as a degree-$(p)$ $V$-erb and a degree-$(q)$ $V$-erb that happens to be on the same cell. This scenario would then correspond to a pair of antiparticles (one of degree-$(p)$ and one of degree-$(q)$) with some interactions. To simplify further, we could identify each erb as several degree-$(1)$ erbs and the system as a set of $p+q$ antiparticle pairs all of degree-$(1)$. By construction of these codes, the effect of error operators on the ground state is equivalent to the same error after composing with $P$-type and $V$-type operators on non-boundary regions. Therefore, two errors that cause the same erbs and are homologous, up to their boundary, act identically and should be considered equivalent errors.

\begin{figure}
    \centering
    \includegraphics[width=.4\linewidth]{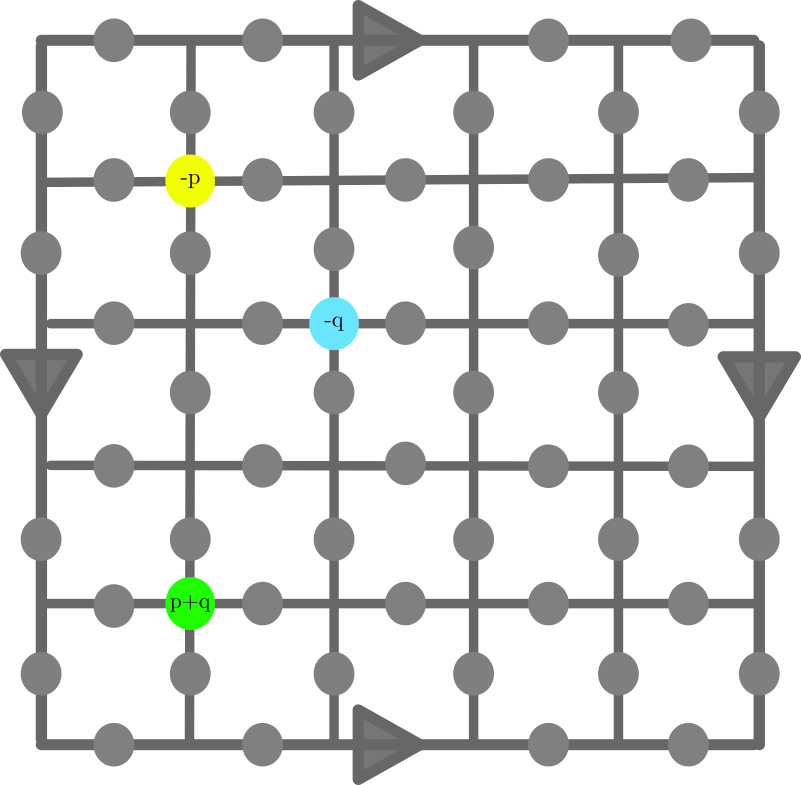}
    \caption{A visual representation of erbs of degrees of $-p,-q$, and  $p+q$ in the toric code}
    \label{fig:p,qerb}
\end{figure}

Note that two degree-$(1)$ $V$-erbs on the same cell are indistinguishable. As these erbs are of the same type, there is no non-trivial braiding or effects of exchanging them around each other; however, the energy of the system may change. For finite $d>3$, higher degree erbs generally lower the total energy of the system, as opposed to multiple separated copies of the same degree-$(1)$ erbs. 

In the toric code, an error is uncorrectable if the particle-antiparticle pair (erb) was annihilated after completing a non-trivial homology class in the manifold. In this way, we can similarly define the annihilation of erbs as when the boundary of the error sub-cellulation is $0$, either by meeting a complementary erb occupying the same boundary or collapsing on its own. For instance, we could have that the $V$-errors of orientation $U$ traced out a closed $k$-chain, $c$. This error operator would act on the basis vector corresponding to the homology class $[c]$, in the code-space, as $X_d^{O(U,c)}$. For a $P$-error with  orientation $U$, this is a closed $(n-k)$-chain in $\widehat{T}$. On a homology class $[c]$ this acts as 

\begin{equation}\prod_{i\in c} Z_d^{O(U|_i,c_i)}=Z_d^{\sum_{i\in c} O(U|_i,c)}\end{equation}

To see that $\sum_{i\in c} O(U|_i,c_i)$ is constant for homologous $c,d$, consider adding the boundary of a $(k+1)$-cell to $c$. Dually, this is an $(n-(k+1))$-cell, $\gamma$, in $\widehat{T}$. As $\hat{\partial} \widehat{c}$ is empty,

\begin{equation}\sum_{i\in c \cap \{\hat{\partial}^\top \gamma\}} O(i,\gamma)=0\end{equation}

Similarly, as $[c]$ is an equivalence class of closed chains, adding the coboundary of a $(k-1)$-cell, $\alpha$, to the error must result in:

\begin{equation}\sum_{i\in c \cap \{\partial^\top \alpha\}} O(i,\alpha)=0\end{equation}

Therefore, this error operator is well-defined on the ground state, regardless of the choice of chain in the equivalence class $[c]$. The above constructions are important for understanding how errors are formed and annihilated, and lead to the following set of questions:

\begin{enumerate}
     \item How does the energy of the system change as erbs change over time?
    \item What are the error-correcting and detecting properties of this code?
    \item What happens if a $V$-erbs moves through a $P$-erb? 
    \item How do the above constructions change for finer triangulations of a PL-manifold?
\end{enumerate}

\subsection{Energy Differences Between Different Erbs}
As previously remarked, the system's energy is the sum of the energy differences due to each violated stabilizer. So, degree-$(1)$ erbs amalgamating on the same cells as a higher degree erb generally decrease the total energy of the system. In the 2D case, erbs are pairs of points in either $T$ or $\widehat{T}$, and this amalgamation is the only way to change the energy of the system while maintaining the total degree of erbs. As the dimension increases, erbs are no longer pairs of points, and therefore, the number of violated stabilizers per erb is no longer constant. Assuming all violated stabilizers are degree-$(1)$, the energy of a lone erb is proportional to the number of cells contained in the erb, which, as the boundary of a sub-pseudo-manifold, can be arbitrary and change without annihilating the erb.

In fact, we can categorize eigenspaces as those with a certain total degree of erbs, which are arranged in a particular way, i.e., whether they are contained as many small isolated single qudit erbs or large erbs bounding a large error. Unlike the toric case, in which an erb would \textit{move} while staying in an eigenspace, the higher homological errors move almost exclusively by changing eigenspaces. For instance, say a $X_{d}$-error occurred on a $k$-cell, $\beta$, which was adjacent to a single $(k-1)$-cell, $\alpha$, containing a V-erb. Assuming that the error and erb were oriented accordingly, the error attaches to the error submanifold of the erb. This error then fixes the stabilizers at $\alpha$, but the erb now contains all the cells in $\partial \beta \setminus \alpha$. This erb has changed shape while changing the energy of the system. By applying multiple-qudit operators, erbs can also \textit{wiggle} to move while remaining in the same eigenspace. This wiggle occurs when the original erb violates $m$ new cells in one portion of the boundary while fixing the same number in a different portion.

\begin{figure}[ht]
    \centering
    \includegraphics[width=.7\linewidth]{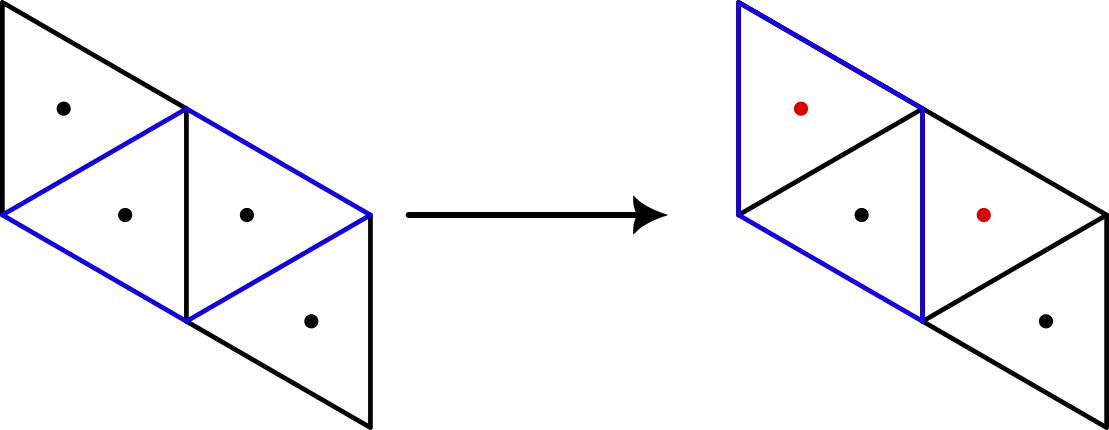}
    \caption{The erb in blue wiggles after a $P$ operation is applied to the two 2-cells marked in red.}
    \label{fig:wiggleerb}
\end{figure}

In different triangulations of the same manifold, two erbs, which are isomorphic as smoothly imbedded submanifolds, may have different triangulations. These erbs would likely have different energies depending on whether they violate greater or fewer cells in these changed triangulations. More precisely, even in the same triangulation, two different erbs can be topologically similar yet correspond to different energy levels in the system, depending on the triangulation of the associated sub-pseudo-manifold. For instance, a $V$-error on two adjacent cells on a fine enough triangulation of any manifold has an erb that is the boundary of an embedded ball and equivalent to a triangulation of $S^{k-1}$. On a slightly altered triangulation, a single $V$-error creates an erb that is also the boundary of an embedded ball and, therefore, a topologically equivalent triangulation of $S^{k-1}$. This simple example shows that while the ground states are topological, the excited states are entirely geometric.

\subsection{Error Correction, and Energy Barriers}

 As for the toric code, the Hamming distance of this code is the size of the smallest homologically non-trivial oriented $k$-sub-pseudo-manifold in $T$ or dual oriented $(n-k)$-sub-cellulation in $\widehat{T}$. This property is directly analogous to the generalized systole problem for $k$ and $(n-k)$-pseudo-manifolds \cite{systole}, in that the minimum of these two systoles gives the hamming distance. 

\subsubsection{Error Correction}
To correct an error, the erb must be annihilated without completing a homology class. When the erb is connected, it is contractible within the triangulation simply by reversing the error that caused it. Therefore, by contracting the erb to a point, if the error is small enough, we can correct such an error. If we knew the error in advance, this would be simple; however, a priori, we can only detect which stabilizers are currently violated. Therefore, for a maximal erb with multiple connected components, each component of the erb must be considered separately. 

\begin{defi}[Error Component/Erp]
    We call a single connected component of an error boundary an error component or erp. The relationship between an erp and an erb of the same type is similar to that between a particle and a particle-antiparticle pair of the same type.
\end{defi}

These erps need not be contractible within the triangulation, for instance, in the QECC encoding the second homology of the torus. If there were a set of $V$-errors going around the central circle, there would be two erps corresponding to these non-contractible paths. These erps can only be annihilated by combining with other erps of the corresponding degree.

\begin{figure}[ht]
    \centering
    \includegraphics[width=.7\linewidth]{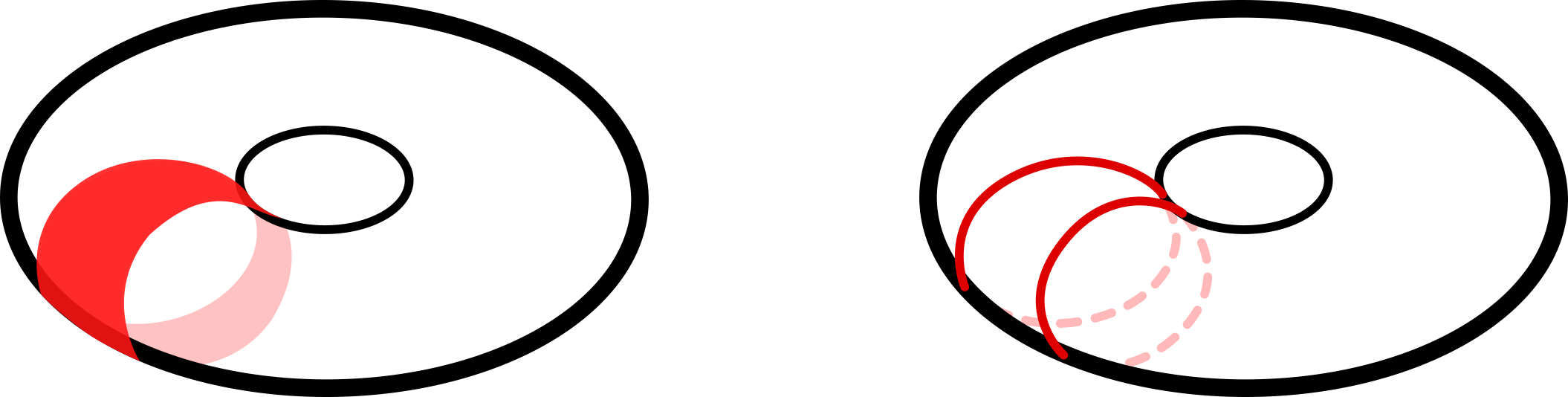}
    \caption{A visual representation of an error on the torus on the left and the corresponding erb on the right. Each circle on the right is a separate error component.}
    \label{fig:erps}
\end{figure}

\subsubsection*{Energy Barriers}
  For erps that are non-contractible within $T,\widehat{T}$ they must correspond to non-trivial $(k-1)/(n-k-1)$-homology classes in $T$ or $\widehat{T}$, respectively. When annihilated, these erps trace out a new error submanifold that, if closed, acts on the code-space as described above. Thus, if the number of broken stabilizers is less than the minimum of the $(k-1)$ and $(n-k-1)$ systole over $T$ and $\widehat{T}$, error correction becomes trivial as all erps are contractible. This correction may corrupt your code space; however, a priori, we can only detect these erps, so we must assume that if an erp is contractible, the corresponding error after this contraction was small enough that it is itself contractible. Therefore, we determine which erps likely correspond to the same erbs only beyond this range so that they can be combined and these errors corrected. 

As amalgamating erbs most often decrease the system's energy, it is energetically favored. Generically, this would lead to erbs shrinking over time; however, this is not always the case. Consider erbs that are triangulated minimal surfaces, i.e., sub-pseudo-manifolds, with the property that any change from a single qudit operator will increase their size, or number of violated stabilizers, even though large changes from an operator acting on many qudits may have a net decrease in their size. For instance, a large triangulated $S^2$ embedded in a triangulation of $\R^3$. This erp is contractible and homologically trivial, yet applying any local operator increases the overall size as it fixes a single stabilizer at the cost of breaking several others. As this is a geometric problem, we can construct arbitrarily bad triangulations in which the number of qudits of the operator that fixes the error must be on the order of the original error before error correction decreases the system's energy. For instance, consider the triangle in $\R^2$ in figure \ref{fig:trianglepic}.

\begin{figure}[ht]
    \centering
    \includegraphics[width=.3\linewidth]{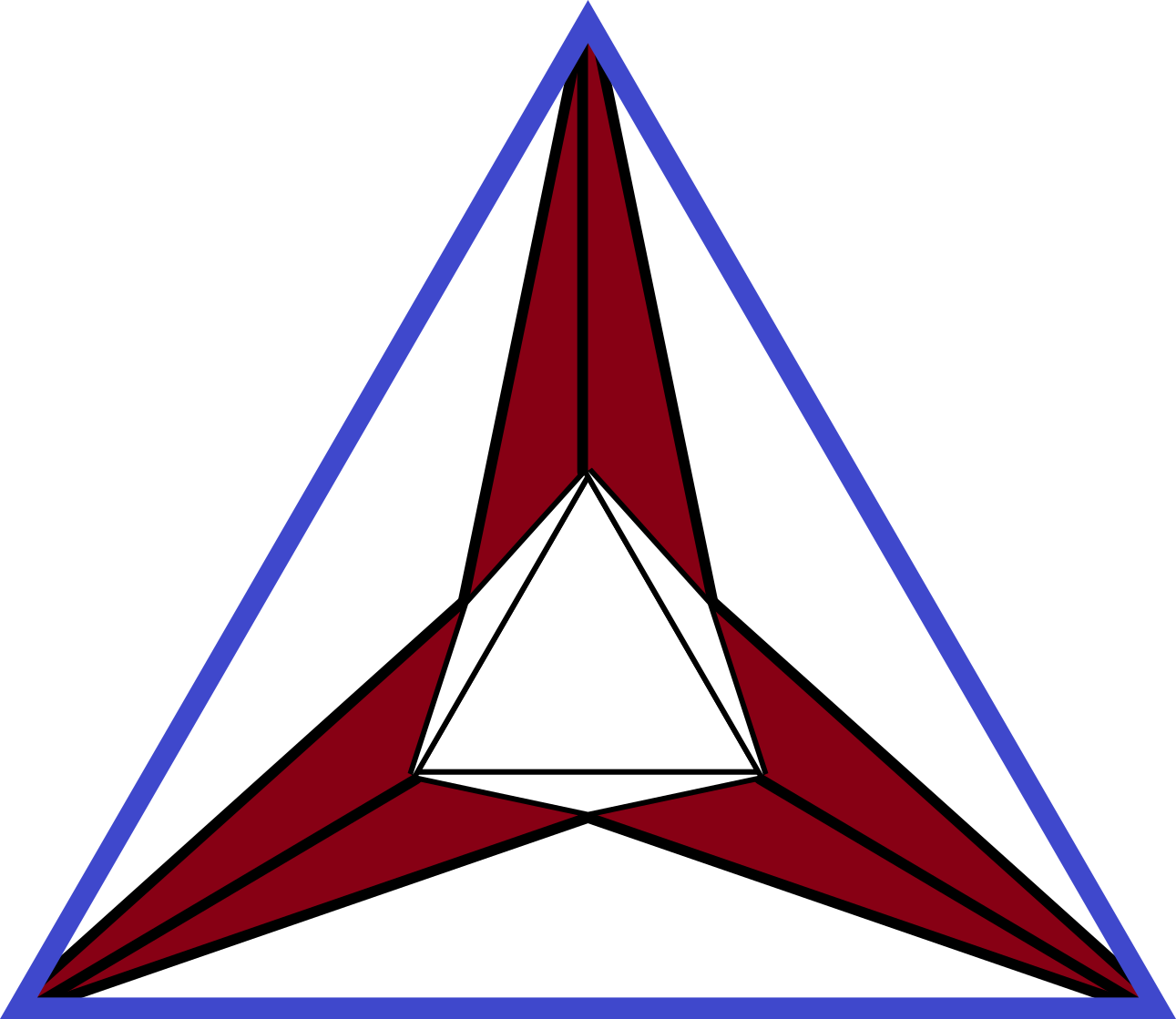}
    \caption{The blue cells are the violated stabilizers. Here, each red triangle is barycentically subdivided, and each subdivision is further subdivided $m$-times. As this is a simplicial complex, there is no smaller triangulation for this error component until it is annihilated at the center of the triangle. Thus, to decrease the number of violated stabilizers, we must apply $O(m)$ local operators.}
    \label{fig:trianglepic}
\end{figure}

We can see that the energy landscape associated with different erps can vary drastically with the choice of triangulation, causing different energy hills and valleys in the overall system \cite{BravyiHaah2011_EnergyLandscape3DSpin}. Since we can choose the triangulation we use, we may be able to obtain a triangulation with a specific energy topography as needed \cite{BravyiHaah2013_QuantumSelfCorrection3DCubicCode}. In general, a lone minimal degree-$(1)$ erp of size $p$ can transition into a minimal higher energy erp of size $q$ by increasing its size first to some $r$ before decreasing down to $q$. Therefore, the activation energy for this event is approximately proportional to $r-p$. The probability of such an event happening over a unit of time at temperature $T$ is roughly proportional to $e^{(r-p)/TR}$, for a fixed constant $R$ \cite{kittel1980thermal}. Thus, at low temperatures, the probability of getting into a maximally disadvantaged state, as in figure \ref{fig:trianglepic}, is very small, and error correction will only need to act on an exponentially small number of additional qudits at a time to get the system into a more favorable state. 

On the other hand, for an error to affect the code-space, its erb must trace out a non-trivial homology class, $l^k$. Therefore, if we consider this error occurring continuously over time, given by some map from $l^k$ to the real line:
\begin{equation}
    f:l^k \rightarrow \R
\end{equation}

Then the erb is a time-like cut of this error, $f^{-1}(t)$. By choosing the best case non-trivial homology class and the best case time-like mapping, we can minimize the smallest erb that must be created for an error $l^k$ to affect the code-space. i.e.

\begin{equation}W_k(T)=\inf_{l^k \in H_k(T)\setminus [0]}(\inf_{f:l^k \rightarrow \R}( \sup_{t\in \R} f^{-1}(t)))\end{equation}

This defines an energy barrier that the system needs to overcome to change the ground state, which must be proportional to the minimum of $W_k(T), W_{n-k}(\widehat{T})$. We have shown that the homological form of these codes creates three relevant geometric problems:

\begin{enumerate}
    \item We want to maximize the minimum of both the $k$-systole in $T$ and the $(n-k)$-systole in $\widehat{T}$ in order to maximize the Hamming distance.
    
    \item We would also want to separately maximize the minimum of both the $(k-1)$-systole in $T$ and $(n-k-1)$-systole in $\widehat{T}$, so as to make error correction trivial.
    
    \item Finally, maximizing the minimum of $W_k(T), W_{n-k}(\widehat{T})$ would require that an erb become very large before affecting the ground state.
\end{enumerate}

\subsection{Braiding Error Components}
 In surface codes, moving excitations around each other can affect the state of the system \cite{kitaev2003fault}. The change in the system when moving two tangled erps versus the same movement on two untangled erps is analogous to this effect of braiding in the surface code case. 
 
 Let $e_1$ and $e_2$ be two erps that change over time. Here, we treat the interval $[0,1]$ as our time dimension and triangulate it fine enough so every change in an erp happens in a unique 1-cell. Treating each erp as an independent sub-pseudo-manifold, these erps each trace out two different sub-pseudo-manifolds in $T\times [0,1]$. We call these maps erp-paths, $e_1(t), e_2(t)$ respectively where $e_1(0)=e_1,e_2(0)=e_2$.
 
 If $e_1$ and $e_2$ were the same form of erp, i.e., both $P$ or both $V$, then as all operators on both commute, there cannot be any non-trivial braiding. Conversely, if not, then as $X_{d}, Z_{d}$ do not commute the erp-paths trace out interacting $k$, and $(n-k)$ dimensional sub-pseudo-manifolds of $T$ and $\widehat{T}$ respectively. As a result of Alexander duality, $n$, and $(n-k)$ dimensional sub-pseudo-manifolds can be non-trivially linked in $S^{n+1}$ \cite{alexanderduality}. Therefore, by embedding $\widehat{T}$ and $T$ into a finer triangulation of $T$, this linking can be well defined in this context. In the paper \cite{dualcellcomplexpaper}, a given cellulation is constructed called the barycentric subdivision of $T$, $\mathcal{T}$, that maps simplicially to both $T$ and $\widehat{T}$. The pullbacks of both erp-paths are then well-defined sub-pseudo-manifolds on this equivalent triangulation of $T$.

 Using the Alexander duality result above, consider a well-defined non-trivial linkage. Undoing this linkage requires moving the $Z$-erp-path through the $X$-erp-path. To do so, consider a secondary time dimension, $s$, so that the overall space is $\mathcal{T}\times [0,1]\times [0,1]$. We can then treat these erp-paths as sub-pseudo-manifolds in $\mathcal{T}\times [0,1]$, and move them in the second time dimension from a tangled $e_1(t),e_2(t)$ to an untangled $e_1'(t),e_2'(t)$. This movement provides us with a pair of isotopies of oriented sub-pseudomanifolds taking $e_1(t)$ to $e_1'(t)$, and $e_2'(t)$ to $e_2'(t)$ $\mathcal{T}\times [0,1]$.

 Moving these erp-paths along these isotopies is done by applications of $X$ and $Z$ operators, so the order in which we move these erp-paths to disentangle does not change the state until the paths reach a cell adjacent to both broken $P$-type and $V$-type stabilizers. Defining the pseudo-manifold of the disentangling as $e_1(t,s),e_2(t,s)\subset \mathcal{T}\times [0,1]\times [0,1] $, the only place that this occurs is where these two sub-pseudo-manifolds intersect. As $e_1(t,s),e_2(t,s)$, are $(k+1),(n-k+1)$-dimensional sub-psuedomanifolds in the $n+2$ dimensional triangulation $\mathcal{T}\times [0,1]\times [0,1]$, the generic intersections are at $0$-dimensional points. Therefore, we can always choose a disentangling for a fine enough triangulation such that each disentangled crossing is at a unique $k$-cell, $\beta_i$, at a unique $t_i$ and $s_i$. The movement of the erps can be defined by the actions of the $X$ and $Z$ operators; changing the type of the crossing at a cell $\beta_i$ is equivalent to switching the order of their application.  Therefore, we can disentangle every crossing by applying a single commutator at each given $k$-cell crossing, and the disentanglement is a sequence of finding such crossings and applying the corresponding commutator.

\begin{figure}[ht]
    \centering
    \includegraphics[width=.5\linewidth]{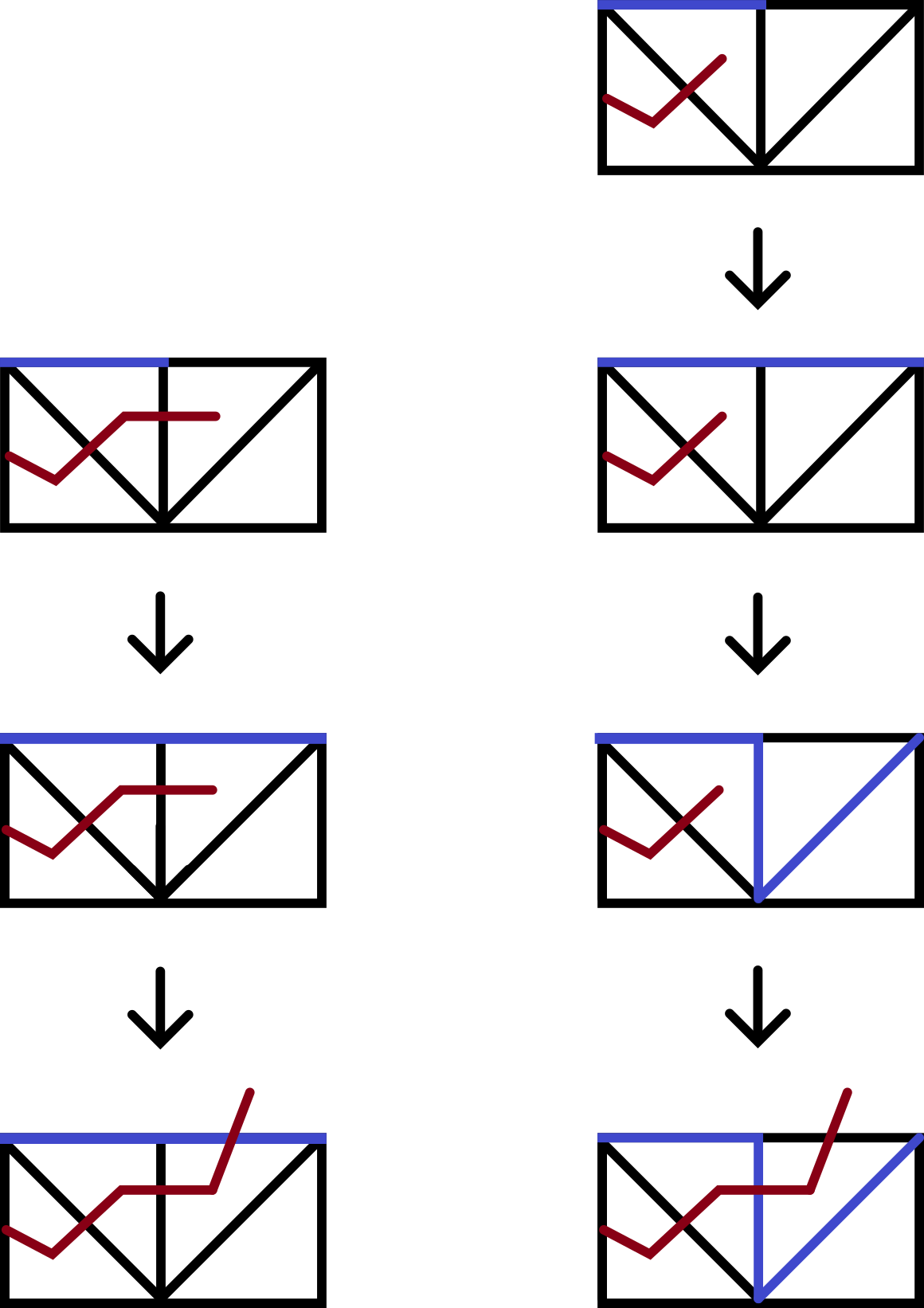}
    \caption{On the left, we have the original erp-paths, and on the right, a $P$ operator is applied in between the intersections. In both cases the $P$ error passes through the $V$-error.}
    \label{fig:brading}
\end{figure}

Note that this correction is identical after applying stabilizers to the code, as if the crossing occurred on a specific $k$-cell $i$ where the $V$-error happened first, then the $P$-error, we can apply a $P$ operator. It would not change the type of the crossing as in figure \ref{fig:brading}. In other words, this change will shift the location and timing of the crossing, yet leave no broken stabilizers. After we do these sequentially, this is equivalent to applying a $P$-operator and, by construction, does not change the state. This example shows that it does not matter which cell we choose to resolve this crossing at. So for some $k$-cell, $\beta$, undoing the crossing at $\beta$ of a $P$-error then a $V$-error is just applying the operator:
\begin{equation}[Z_{d,\beta}^{O(u_1|_\beta)},X_{d,\beta}^{O(u_1|_\beta)}]=\zeta_d^{O(u_1|_\beta,\beta)O(u_2|_\beta,\beta)}\end{equation}

Reversing this process is explicitly entangling these two erp-paths and would act by the inverse on the Hilbert space as:

\begin{equation}[X_{d,\beta}^{O(u_2|_\beta)},Z_{d,\beta}^{O(u_1|_\beta)}]=[Z_{d,\beta}^{O(u_1|_\beta)},X_{d,\beta}^{O(u_1|_\beta)}]^{-1}\end{equation}

There are many cases in which more than one move is needed to detangle two of these submanifolds fully. i.e., in the standard toric code, a double link with crossing number 2 would need to be moved past the other string twice to detangle. However, the resultant operator found is invariant under any choice of $k$-cell within a small neighborhood, implying that the effect of the unlinking on the Hilbert space depends only on the topology of the disentangling chosen. Specifically, we have shown that the net phase equals the signed intersection number of the two space-time world-volumes.

\section{Encoding Characteristic Classes}

Now that we have fully developed codes that encode homology and cohomology, we have the tools to encode sections of fiber bundles and their obstruction classes. A fiber bundle is a triple of spaces $\{F, E, B\}$, referred to as the fiber space, total space, and base space respectively, and a continuous surjective map $p: E\rightarrow B$, such that for any $b\in B$, there is a neighborhood $U_b$ such that $p^{-1}(U_b)$ is homeomorphic to $U_b\times F$ \cite{milnor1974characteristic}. This homeomorphism, $\phi_b$, called the trivialization of $U_b$, must agree with the projection i.e:

\begin{equation}p: p^{-1}(U_b) \rightarrow U_b = proj_1\circ \phi_b: p^{-1}(U_b) \rightarrow U_b\times F \rightarrow U_b\end{equation}

By construction, $p^{-1}(b)$ must be homeomorphic to $F$, and so a fiber bundle can be thought of as a subtle way to replace each point in the base space with a copy of the fiber. A section of a fiber bundle is a continuous map $\sigma: B\rightarrow E $ such that $p \circ \sigma = I$. To construct a section of a bundle, we can start by constructing a section of the $0$-skeleton of the base space, then extend it to the $1$-skeleton, and so on, until we obtain a section of the $n$-skeleton, which is a full section of the bundle. Define $\sigma_{k-1}:B_{k-1}\rightarrow E$, as a section of the $(k-1)$-skeleton that can be extended to the $k$-skeleton. The next step would then be to choose an extension of the $k$-skeleton that we can further extend to the $(k+1)$-skeleton. For a $(k+1)$-cell, $\gamma$, and extension of $\sigma_{k-1}$, $\tilde{\sigma}_k$, we have that:

\begin{equation}\tilde{\sigma}_k|_{\partial\gamma}: \partial \gamma \rightarrow p^{-1}(\gamma) \cong \gamma \times F\end{equation}

As $\partial \gamma$ is homeomorphic to $S^{k}$ and $\gamma \times F$ is homotopically equivalent to $F$, this defines a map from each $k$-cell to an element in $\pi_{k}(F)$. By definition, the section can be extended over $\gamma$ if the map restricted to $\partial \gamma$ is homotopic to $0$. Therefore there is a $\pi_{k}(F)$ simplicial $(k+1)$-cochain of $B$, $Ob_{\tilde{\sigma}_{k}}$, that determines 
whether said section can be extended only if $Ob_{\sigma_{k}}=0$. The differential applied to $Ob_{\tilde{\sigma}_k}$ restricted to a $(k+2)$-cell $\eta$, is given by: 
\begin{equation}d (Ob_{\sigma_{k}})(\eta)=\sum_{i\in\partial\eta} Ob_{\tilde{\sigma}_k}(i)=\prod_{i\in\partial\eta}\prod_{j\in \partial i} \tilde{\sigma}_k|_{j^{O(i,j)}}\end{equation} 

Where the product on the right-hand side is the concatenation of maps and $ \tilde{\sigma}_{k}|_{j^{O(i,j)}}$ is the section $\tilde{\sigma}_{k}$ restricted to $k$-cell, $j$, with the orientation given by the orientation of $i$ with respect to $j$. However the two $(k+1)$-cells, $i,i'$ adjacent to the $k$-cell $j$, must have opposite orientations. As higher homotopy groups are abelian, we can permute these maps so that each cancels with their inverse \cite{hatcher2002algebraic}, leading to:

\begin{equation}d(Ob_{\tilde{\sigma}_k})(\eta)=\tilde{\sigma}_k|_{i^1}\circ\left ( \prod_{j\in \partial \partial \eta\setminus i} \tilde{\sigma}_k|_{j^1}\circ \tilde{\sigma}_k|_{j^{-1}}\right)\circ \tilde{\sigma}_k|_{i^-1}=0\end{equation}

 Therefore as $dOb_{\sigma_k}(\eta)=0$ for all $\eta$, $Ob$ must a cocycle.
 
 Given a different section of the $k$-skeleton that extends $\sigma$, ${\tilde{\sigma}_k}'$, that disagrees with $\tilde{\sigma}_k$ only on the interior of a single $k$-cell, $\beta$. Up to homotopy, this disagreement is given by a class of maps from $D^k$ to $F$ that fixes $\partial D^k$. By concatenating the two maps along their shared boundary, up to homotopy, we construct a map in $\pi_{k}(F)$. More generally, up to homotopy, each $k$-chain of $B$ over $\pi_{k}(F)$ defines each ${\tilde{\sigma}_k}'$ that extends the section $\sigma_{k-1}$. In terms of our example if this element of $\pi_k(F)$ is given by $r$ we have that:
 
\begin{equation}Ob_{{\tilde{\sigma}_k}'}+rd\beta= Ob_{\Tilde{\sigma}_k}\end{equation} 

Therefore, define $Obc_{\sigma_{k-1}}\in H^{k+1}(B,\pi_k(F))$ as the equivalence class of all such cocycles $Ob$, over all sections, $\tilde{\sigma}_k$ that extend $\sigma_{k-1}$. This class is $0$ if there is a section on the $(k+1)$-skeleton that extends $\sigma_{k-1}$. On the other hand, as these sections have been defined only up to homotopy, if we allow $p\circ \sigma_{k+1}$ to be homotopic to the identity as opposed to strictly equivalent to the identity, then this class is $0$ if and only if there is a section on the $(k+1)$-skeleton that extends $\sigma_{k-1}$  \cite{milnor1974characteristic}. For triangulable manifolds, we can triangulate all maps this way; so, this is an invariant for the existence of a section up to the $(k+1)$-skeleton of $B$.

\subsection{Example: The Hairy Ball Theorem}

The Euler class of $S^2$, whose non-triviality is also known as the hairy ball theorem, is a classic first example for a discussion of characteristic classes. We will construct this example for the cellulation of the sphere given by a square box, and consider its unit tangent bundle, whose fiber is $S^1$ as in figure \ref{fig:euler class}. To each of the $6$ faces, labeled front, back, left, right, top, and bottom, of the box, we can define the open sets $\{U_F, U_{Ba},\dots U_{Bo}\}$ that are each a small collar neighborhood larger than their corresponding faces. The trivialization of a given $U_i$ is given by:
\begin{equation}p: p^{-1}(U_i) \rightarrow U_i = proj_1\circ \phi_i: p^{-1}(U_i) \rightarrow U_i\times F \rightarrow U_i\end{equation}

Over each open set $U_i$ we can define a section, $s_i$, to $U_i\times S$ given by $\phi_i(s_i(x))=x\times 0$, for $0\in S^1$. As these trivialization may differ, for some $x$ in the intersection $U_i\cap U_j$, there is some $\theta_x\in S^1$ such that: 
\begin{equation}
    \phi_j(\phi_i^{-1}(x\times 1))=x\times \theta_x
\end{equation}
As the intersection of any two $U_i, U_j$ is either empty or a disk, we can choose the trivializations, $\phi_i$, so that going from one trivialization to another is a rotation of the fiber by a constant $\theta$. For each intersecting $U_i,U_j$, this defines a fixed $\theta_{i,j}\in S^1$ that describes this difference in these trivializations.

For this example, lets orient the $\theta$ so that the forward, $F$; left, $L$; right, $R$; and back, $Ba$ open sets pairwise have $\theta_{i,j}=0$, i.e.,
\begin{equation}\theta_{F,L}=\theta_{L,Ba}=\theta_{Ba,R}=\theta_{R,F}=0\end{equation}

The forward face $F$ will also have no rotation with the top, $T$ and bottom, $Bo$ faces,  i.e., 
\begin{equation}\theta_{F,T}=\theta_{F,Bo}=0\end{equation}

These trivializations correspond to the unfolding of the box given by figure \ref{fig:boxunfold}, which results in:
\begin{equation}\begin{split}
    &\theta_{L,T}=\pi/2,\ \theta_{Ba,T}=\pi,\ \theta_{R,T}=3\pi/2,\ \\ 
    &\theta_{L,Bo}=3\pi/2,\ \theta_{Ba,Bo}=\pi,\ \theta_{R,Bo}=\pi/2\ 
\end{split}\end{equation}

\begin{figure}
    \centering
    \includegraphics[width=.5\linewidth]{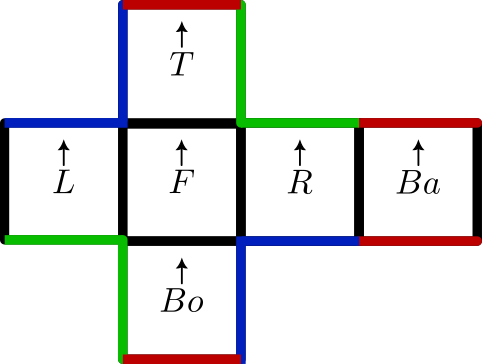}
    \caption{The unfolding of the box described above, the arrow points to the group element $0\in S^1$, and the colors correspond to edges that are pasted together to make the cube.}
    \label{fig:boxunfold}
\end{figure}

By the hairy ball theorem, there is no way to choose a continuous section over each face that agrees with the other chosen sections on each intersection after applying the corresponding transitions $\theta_{i,j}$. 

\subsection*{Making a Topological Code}
We start by constructing a section over each $0$-cell. As $S^1$ is path connected, and all maps discussed are equivalent up to homotopy, any choice is equivalent. So, we can place a copy of $\C[\Z_4]$ on each vertex is a discrete form of $\C[S^1]$, i.e., $\ket{a}$ corresponds to $\ket{\pi a/2}$, and fix the quantum data $\ket{0}\in \C^4\cong \C[\Z_4]$ on each of said vertices. However, as this point on the fiber is well defined only up to a choice of trivialization this association of $\ket{a}$ corresponding to $\ket{\pi a/2}$ is defined over a chosen face as in figure \ref{fig:boxtriv} and on adjacent faces is adjusted by $\theta_{i,j}$. For instance, if a vertex had $\ket{a}$ correspond to $\ket{\pi a/2}$ for a face $i$ then over a face $j$, $\ket{a}$ would correspond to $\ket{\pi a/2+ \theta_{i,j}}$.

\begin{figure}[ht]
     \centering
     \includegraphics[width=.5\linewidth]{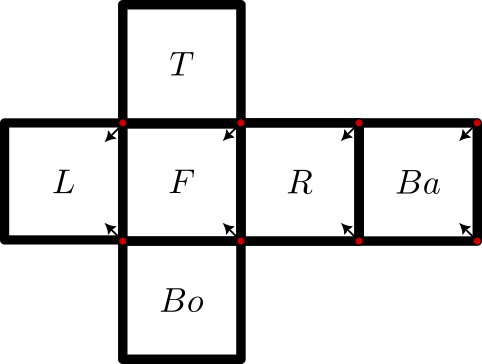}
     \caption{In this unfolding, each $0$-cell is given by a red dot and points to the center of the face whose trivializations it is using to define $\ket{i}$}
     \label{fig:boxtriv}
 \end{figure}

\subsubsection*{Constructing Corresponding $P$-type Operators}
To extend this section, for every $1$-cell, we must choose a section that agrees with the pair of $0$-cells on its boundary. We define $[I,S^1]|_{\partial I}\cong \Z$ as this class of maps from the interval to $S^1$ that fixes the endpoints of the interval, up to homotopy. We can set this isomorphism to count the number of clockwise rotations from the base point to the end point along the $1$-cell. On the intersection of two different open sets, $U_i,U_j$, the map taking one trivialization to the other, $\phi_j\circ \phi_i^{-1}$ must induce an isomorphism on the homotopy classes $(\phi_j\circ \phi_i^{-1})^*$. We can encode the number of clockwise rotations on each 1-cell by associating with each one a copy of $\C[\Z]\cong\C[[I,S^1]|_{\partial I}]$.

For the encoding to be well defined over both trivializations, if $\ket{a}$ is encoded in $i$'s trivialization corresponding to $a\in [I,S^1]$, then in $j$'s trivialization $\ket{a}$ is interpreted as $\ket{(\phi_j\circ \phi_i^{-1})^*(a)}$. This isomorphism is defined by the path that $\ket{0}$ describes in each trivialization. A consistent way to define $\ket{0}$ is by choosing a generic regular value point, $c$ in $S^1\setminus \{0,\pi/2,\pi, 3\pi/2 \}$, for example $c=\pi/4$, and letting $\ket{0}$ correspond to the path between the endpoints that does not go through $\pi/4$ over a given trivialization. This difference in trivialization can cause one of the $0$-cells to pass through $\pi/4$ after applying the rotation $\theta_{i,j}$. In this case, the isomorphism $(\phi_j\circ \phi_i^{-1})^*$ adds or subtracts one accordingly as in figure \ref{fig:circle passing}.
 
 \begin{figure}[ht]
     \centering
     \includegraphics[width=\linewidth]{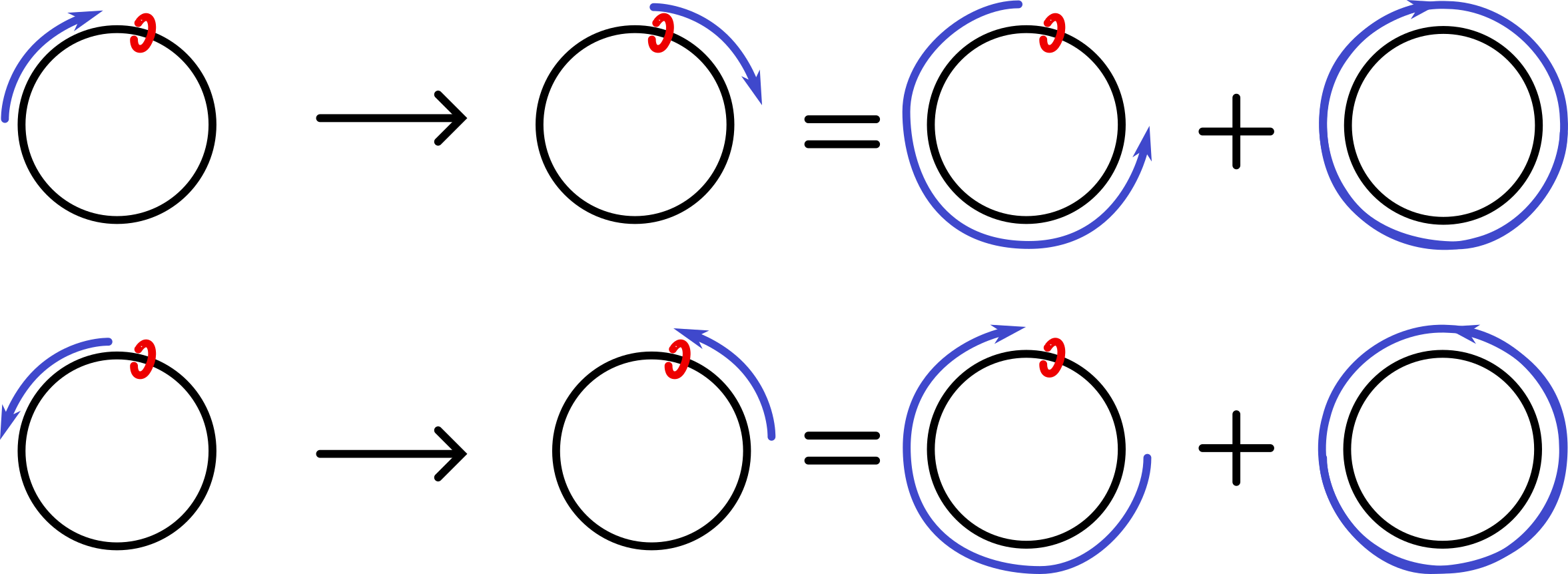}
     \caption{This determines the change in the encoding given from changing the trivialization to one passing through the chosen point given by the red circle.}
     \label{fig:circle passing}
 \end{figure}
 
In this example, after fixing each $0$-cell, we can define the function $t=(\phi_j\circ \phi_i^{-1})^*$ as given in \eqref{rotation_in_cube}. This adjustment then takes a $1$-cell with data $\ket{0}$ and treats it as $\ket{a}$ for the face $i$, but also as $\ket{a+\tau(\theta_{i,j})}$ for face $j$. 

 \begin{equation} \label{rotation_in_cube}
 \tau^j_i=
 \begin{cases}
  -1 & \{i,j\}=\{F,T\},\{F,B0\}\\
   1 & \{i,j\}=\{T,F\},\{Bo,F\}\\
   0 & \text{ otherwise }
 \end{cases}
 \end{equation}

The number of rotations along the boundary of a face is then equal to the sum of the number of full rotations along each edge of its boundary according to that face's trivialization. For instance, without loss of generality, say that the trivialization on face $i$ agrees with the adjacent $0$ and $1$-cell data. Then the cohomological $P_\Z$ operator acts as the identity if and only if each section that is encoded can be extended over the interior. For a face $j$ that intersects with $i$ such that $\theta_{i,j}$ is non-zero, the $Z_\Z$ operator associated with the edge they share in the necessary $P_{\Z,j}$ operator is replaced by:

 \begin{equation} \tilde{Z}^j_i=
 (X_\Z)^{-\tau^j_i} \otimes Z_\Z \otimes (X_\Z)^{\tau^j_i}\end{equation}

Which, as an operator, acts as:

\begin{equation}\tilde{Z}_ij(\ket{l})=e^{2\pi i r (l+ \tau^j_i)}\ket{l}\end{equation}

As for cohomological codes, when the orientation of a $2$-cell and an edge on its boundary disagree, we invert these $Z$-type operators, as the corresponding rotation on that edge is going in the opposite direction as to how it is being added up for that $2$-cell. We can then construct cohomology $P_\Z$ operators out of these adjusted operators as was done in the cohomological codes above. The sections encoded are then extendable over the 2-skeleton if and only if all of the $P$-type operators are satisfied at the same time.

\subsubsection*{Constructing Corresponding $V$-type Operators}
In the construction, this section of the $k$-skeleton was not defined up to homotopy equivalence. To find a homotopically equivalent section of the $0$-skeleton, we can rotate the section over a $0$-cell, from $\ket{0}$ to $\ket{1}$, $\ket{1}$ to $\ket{2}$ etc. adjusting the corresponding $1$-cells so that the homotopy class of rotations along each cell is the same. After one full clockwise rotation of a $0$-cell, the state it encodes is still $\ket{0}$, and the section as a whole has not changed up to homotopy. Yet, each adjacent $1$-cell qudit has increased or decreased the number of rotations depending on whether this was the basepoint or the endpoint of the corresponding edge. This change in the $1$-cells is equivalent to applying the cohomology $V_{\Z}$ operator on the corresponding $0$-cell. Therefore, if the $V_{\Z}$ operators act as the identity, this corresponds to an explicit section of the $1$-skeleton up to homotopy.

We can define a code with the adjusted operators as stabilizers; however, not all $P_\Z$ operators in this code can be satisfied simultaneously. The broken stabilizers then correspond to an obstruction $2$-cocycle, which is an explicit element of the obstruction class.
\begin{figure}
    \centering
    \includegraphics[width=.3\linewidth]{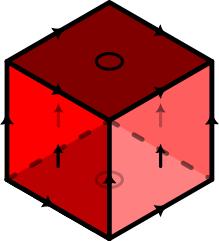}
    \caption{A visual representation of the obstruction chain given by errors on the top and bottom face.}
    \label{fig:euler class}
\end{figure}
With slight adjustments, this method can be applied to any triangulation of $S^2$. Further note that changing the chosen trivializations would shift the encoded cochains by a coboundary and therefore not change the structure of the erbs. We could also decide not to fix the $0$-cell data, and have a rotation stabilizer $A_\alpha$ such that $A_\alpha^4=V_\alpha$. The $P_\Z$ would need to be further adjusted by using generalized Toffoli operators to adjust the $\tilde{Z}_\Z$ operators instead of the $X_\Z$ operators. However, as any choice for this would not change the homotopy class of the section, the $0$-cells would not encode any data.

\subsection{Example: P-type operators on $S^k$ Bundles}
Before giving the general construction, we will provide an explicit method for constructing $P$-type operators for obstructions to sections of bundles with fiber $S^k$, i.e., $S^k$-bundles. In the circle bundle case, when switching from one trivialization to another, we can find the effect of the induced map $(\phi_j\circ \phi_i^{-1})^*$ on an edge by finding out which endpoint the fixed point $c$ passed through, as we did in equation \eqref{rotation_in_cube}. For $S^k$ bundles, this same strategy requires a more subtle consideration. As before, we define fixed sections of dimension less than $k$ arbitrarily as $\pi_i(S^k)=0$ for all $i<k$. We place $\C[\Z]$ on each $k$-cell and define $\ket{0}$ as a section, up to homotopy, that agrees with the sections on the $k-1$-cells such that a given point $c$ is not contained in the fiber. Then the $P$-type operator must be adjusted to account for whether $c$ is not in the fiber over $\beta$ after changing trivializations and whether this contributes a positive or negative adjustment to the encoded data. 

To encode this adjustment, consider all $(k-1)$-cells bounding the $k$-cell, $\beta$, let $\rho(t)$ be a continuous map from $[0,1]$ to the structure group of the bundle such that $\rho(0)=I,\rho(1)=\rho^{\beta}_{\gamma}$ is the operator that takes the trivialization on $k$-cell $\beta$ to $k$-cell $\gamma$. For instance, in the circle bundle case, $\rho$ is a path that twists the fiber from one trivialization to another as $t$ goes from $0$ to $1$. 

Then we define $\rho^\beta_{k-1}(c)^{-1}$ as the set of $(k-1)$-cells that generically intersect $\rho(t)^{-1}(c)$ for some $t$. This, in effect, computes the
degree of the transition on the boundary sphere. If the number of $(k-1)$-cells in $\partial \beta$ that intersect $\rho(t)^{-1}(c)$ is odd, then $\rho^{\beta}_{\gamma}$ contains $c$ as it entered and left an odd number of times. If the orientation of the section of $\beta$ agrees with the orientation of the fiber, once it contains $c$, it will invert, so the value in the corresponding qudit needs to increase by $1$. Conversely, if the orientations disagree, this value must decrease by $1$. If, for the boundary of some $k$-cell, there is an even number of crossings, then $c$ leaves and re-enters an even number of times and ends up not contained in the section of the $k$-cell. More precisely, we can define and apply the following adjustment maps:

\begin{equation}\tau^{\beta}_
 {\gamma_1}=
   \begin{cases} 
      O(\beta,F) & \rho^\beta_{k-1}(c)^{-1}\text{ is odd} \\
      0 & \rho^\beta_{k-1}(c)^{-1}\text{ is even}\\
   \end{cases}
\end{equation}

As we fixed the data on the $(k-1)$-skeleton, the adjustment given by the above equation can be found for each nonempty pairwise intersection and applied as we did for the circle bundle. This adjustment then allows us to construct the $P$-type operators to encode whether the section is extendable over each $(k+1)$-cell.

\subsection{The Obstruction Class Hamiltonian}
    If we were to use the more physical version of these qudits as described in section \ref{annoying:integer}, we can construct the following Hamiltonian: 
    
\begin{equation}H_{Obstr,r}(v)=\sum_{\gamma\in T_2} (2-\tilde{P}_{\gamma}-\tilde{P}_{\gamma}^{-1})(v)+\sum_{\alpha \in T_0} (2-V_{\alpha,r}-V_{\alpha,r}^\dagger)(v)\end{equation}

For the trivial fiber bundle, i.e., when each $f^{\gamma_2}_{\gamma_1}=0$, this code encodes cohomology. So, for general fiber bundles, the lowest-energy space corresponds to an encoding of cohomology with an unfixable manufactured erp contained within the code, which is a representative of the obstruction cochain.

\subsection{Encoding the Obstruction Class over any Fiber Bundle.}

To construct a section of a fiber bundle $\{F, E, B\}$, we will start by creating a section over each $0$-cell, which, as discussed above, we are free to do arbitrarily, if and only if the fiber is path-connected. Next, we must show that for a given section on $k$-cells, we can extend it to a section on $(k+1)$-cells of the triangulation of $B$ only if the obstruction cochain associated with it is trivial. For this purpose, we will assume both path-connectedness of the fiber and that the $k$th homotopy group of the fiber is both abelian and finitely generated.

Given a section over the $k$-skeleton, $\sigma_{k}$, our encoding must determine whether we can extend this to a section of the $(k+1)$-skeleton. Thus, for each $k$-cell, we place a qudit that will encode the homotopy class of its map on the fiber. As the homotopy class of maps from $D^k$ to $F$ that fix the boundary of $D^k$, $[D^k,F]|_{\partial D^k}$, is isomorphic to $\pi_k(F)$, this qudit can be given by a basis vector in $\C[\pi_k(F)]$ on every $k$-cell. 

To make this encoding well-defined, for each $(k+1)$-cell, $\gamma_1$, with $k$-cell, $\beta$, in its boundary, we choose a map of $\beta$ to the fiber that agrees with the section on the $(k-1)$-skeleton, $f^{\gamma_1}_\beta:\beta \rightarrow F$. This association of $f^{\gamma_1}_\beta$ for $\gamma_1$ will allow us to define arbitrary sections of $\beta$, $g^{\gamma_1}_\beta: \beta \rightarrow F$, that agree with $f_\beta$ on its boundary but are not necessarily homotopically equivalent. As $\beta$ is a $k$-cell, we can paste these two maps together along the boundary of $\beta$ to construct the map:
\begin{equation}g\sqcup_{\partial \beta} h:S^k\rightarrow F\end{equation}

To define the element in $\pi_k(F)$ corresponding to this map, this needs to be a map of oriented manifolds. We choose an orientation on the mapped $S^k$ given by the orientation for $\beta$ on the $g$ hemisphere and the reverse orientation on $\beta$ for the $f$ hemisphere. As these orientations are compatible, we can glue the hemispheres together to create an oriented $S^k$, and define an element $j\in \pi_k(F)$ corresponding to the homotopy class of this map. We can then set $\ket{j}$ on $\beta$'s qudit to encode the map $g^{\gamma_1}_\beta$ up to homotopy equivalence.

In a different trivialization corresponding to $(k+1)$-cell $\gamma_2$ with $\beta$ in its boundary, there is a transition function $t^{\gamma_2}_{\gamma_1}$ applied to $g^{\gamma_1}_\beta$, and $f^{\gamma_1}_\beta$.
 The encoded qudit is well-defined under these different trivializations, as the transition functions are homeomorphisms of the fiber and, therefore, define an isomorphism on all homotopy groups. We can then interpret $\ket{i}$ as $i\in \pi_k(F)$ for $\gamma_1$ as $(t^{\gamma_2}_{\gamma_1})^*(i)$ for $\gamma_2$. As for the circle case, this allows us to define adjustments $h^{\gamma_2}_{\gamma_1}$ such that $\ket{i}$ for $\gamma_1$ is interpreted as $\ket{i+h^{\gamma_2}_{\gamma_1}}$ for $\gamma_2$. Thus, there is a well-defined encoding of the section over $\beta$ and its interpretation in different trivializations. This construction differs here from the $k$-sphere case, as we do not specify a specific isomorphism between each given $\C[\pi_k(F)]$ and the isomorphic vector space $ \C[D^k,F]|_{\partial D^k}$. Instead, we must calculate this for each fiber bundle and each choice of trivializations.

Now that we have encoded the sections on the $k$-skeleton, we must encode whether this section can be extended over the $(k+1)$-skeleton of the base manifold. For any $(k+1)$-cell $\gamma$, we first adjust all data to agree with $\gamma$'s trivializations. Then, without loss of generality, we can invert the data on the $k$-cells in its boundary that are not in the induced orientations from $\gamma$, i.e., $\ket{a}$ will be interpreted as $\ket{-a}$ if the orientation on the corresponding cell disagrees with its induced orientation from $\gamma$. So we can assume that all boundary cells have orientations induced by $\gamma$. Next, along the boundary of $\gamma$ for each pair of adjacent $k$-cells $\beta_1,\beta_2$ we paste each given $f^{\gamma_1}_{\beta_1}$ map to the adjacent $f^{\gamma_1}_{\beta_2}$ map along the common boundary $\partial \beta_1\cap \partial \beta_2$. This results in a map from $S^k$ to $F$ that would be the given section if all the qudits encoded $\ket{0}$. Therefore, when summing the values along the boundary and this value, we get the homotopy class of the $g$ maps glued along common boundaries, as in figure \ref{fig:pastingforcharclass}.

\begin{figure}[ht]
    \centering
    \includegraphics[width=.5\linewidth]{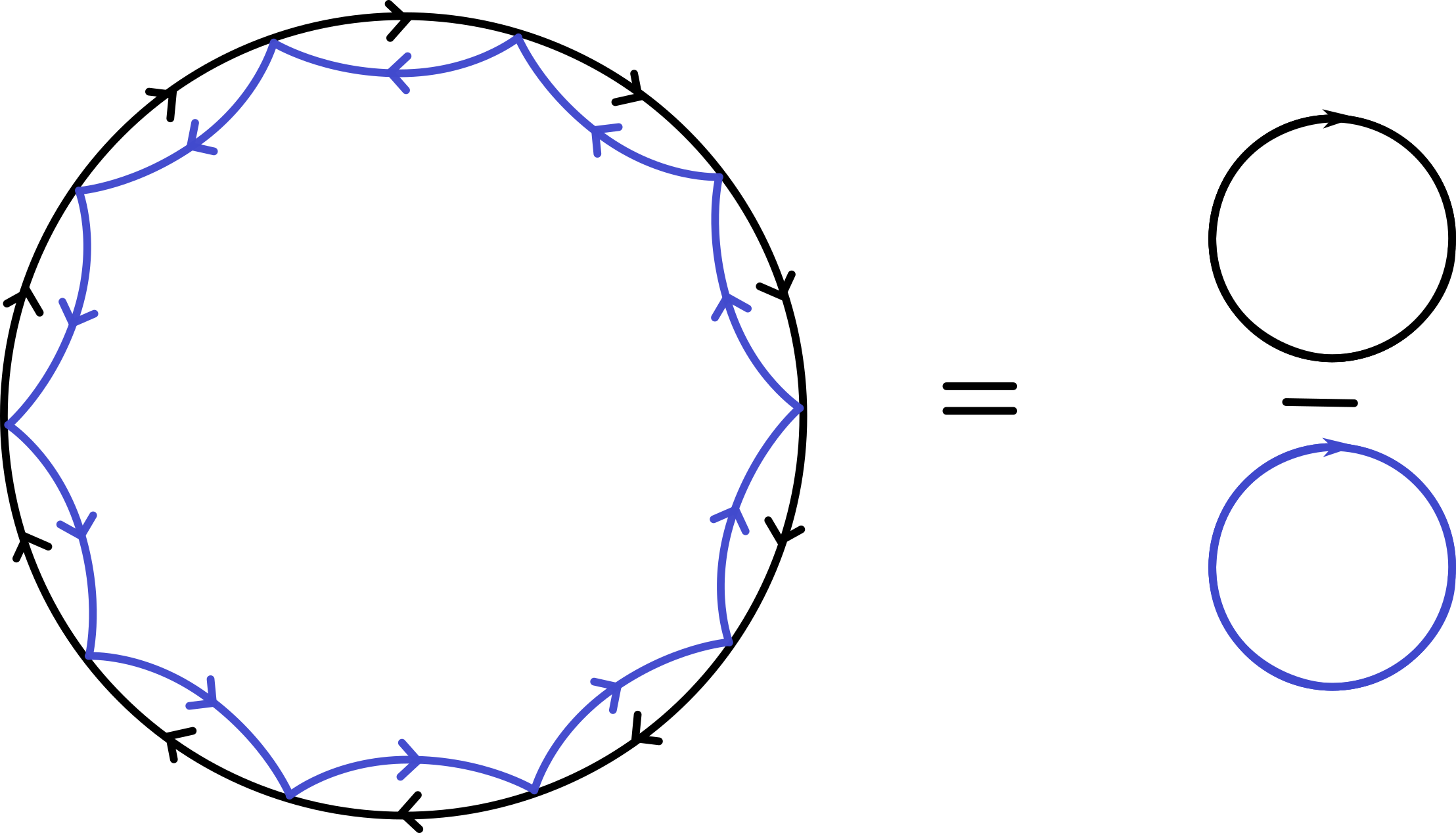}
    \caption{A visual representation of the cutting and gluing of the maps. The blue lines are the $f$ maps, and the black lines are the encoded $g$ maps.}
    \label{fig:pastingforcharclass}
\end{figure}

To construct this homotopy explicitly, consider the concatenation of two adjacent connected contractible $k$-chains $\beta_1,\beta_2$ with shared, connected boundary $(k-1)$-chain $\alpha$. Without loss of generality, take a point in the image of $\alpha$ as the base point for the homotopy groups. As all maps $f^\gamma_{\beta_1},f^\gamma_{\beta_2},g^\gamma_{\beta_1},g^\gamma_{\beta_2}$, are maps from the disk that agree on $\alpha$, we can reinterpret the homotopy class represented on $\beta_1,\beta_2$ as the equivalent homotopy relative to $\alpha$, given that $\alpha$ is contractible in $S^k$. We can then see that the sum of the homotopies, as a map from $S^k$ to $F$, extends through $S^k\sqcup_\alpha S^k$ by setting a collar around the equator of $S^k$. So we have that as a map, the homotopy group element: $f^\gamma_{\beta_1}\sqcup_{\partial \beta_1}g^\gamma_{\beta_1}+  f^\gamma_{\beta_2}\sqcup_{\partial \beta_2}g^\gamma_{\beta_2}$ factors through:
\begin{equation}
\phi:S^k\rightarrow (f^\gamma_{_{\beta_1}}\sqcup_{\partial \beta_1}g^\gamma_{_{\beta_1}})\sqcup_\alpha  (f^\gamma_{\beta_2}\sqcup_{\partial \beta_2}g^\gamma_{\beta_2})\rightarrow F\end{equation}

Then, as a map from $S^k$ to $F$, the equator is sent to $\alpha$. However, by cutting this map from north to south along the intersection of $\beta_1,\beta_2$, this map is homotopically equivalent to $\bar{f}^\gamma_{\beta_1\cup \beta_2}\sqcup_{_{\partial \beta_1 \cup \partial \beta_2}} g^\gamma_{\beta_1\cup \beta_2}$, as seen in figure \ref{fig:fiber bundle pasting}.

\begin{figure}[ht]
    \centering
    \includegraphics[width=\linewidth]{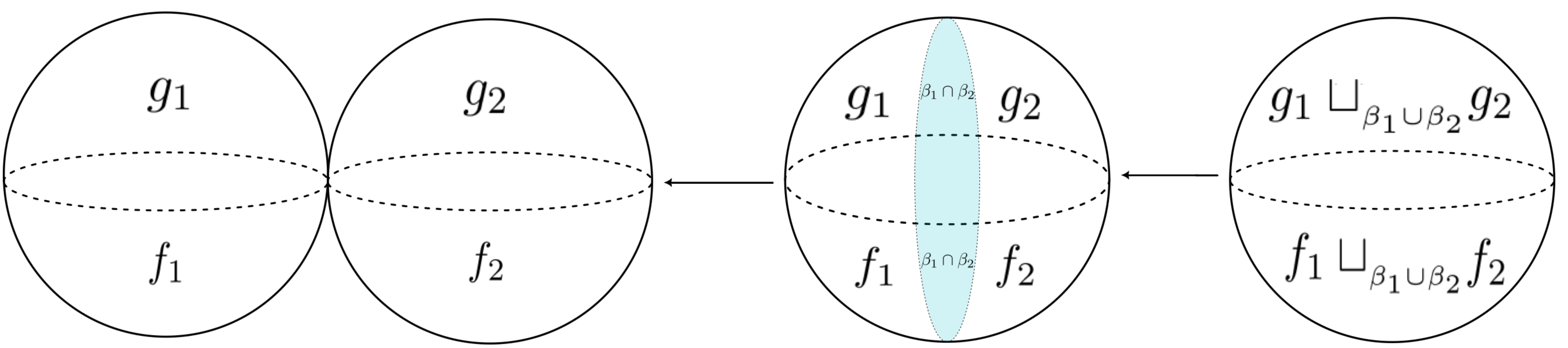}
    \caption{A visual representation of the pasting of the maps $f$ and $g$}
    \label{fig:fiber bundle pasting}
\end{figure}

Here the notation $f^\gamma_{\beta_1\cup \beta_2}$ is the pasting of $f^\gamma_{\beta_1}$ and $f^\gamma_{\beta_2}$ along their common boundary. The $\bar{f}$ comes from the fact that the orientations for the $f$ maps were opposite the orientations of $\beta_1,\beta_2$. By iterating this over all $k$-cells in the boundary of $\gamma_1$ we can combine the $g$ and $f$ parts of the map as they go around $\partial \gamma$ ending up with a map:
\begin{equation}g^\gamma_{\partial\gamma}\lor \bar{f}^\gamma_{\partial\gamma} :S^k\lor S^k \rightarrow F\end{equation}

By composing with the map $f$ from $S^k$ to $F$, we get a map that is homotopically equivalent to $g$. In summary, this corresponds to adding the corresponding element for $f^\gamma_{\partial \gamma}\in \pi_k(F)$ to the sum of the encoded states on the boundary of $\gamma_1$. As $\pi_k(F)\simeq \Z^{l_0}\bigotimes_{i=1}^m \Z_{d_i}^{l_i}$, we can write $f^\gamma_{\partial \gamma}$ as $\bigoplus^m_{i=0}\bigoplus_{j=0}^{l_i}f_{i,j}$ and the corresponding $P$-type operator, $P^{Obstr}_{\gamma_1}$, defined as: 

\begin{equation}P^{Obstr}_{\gamma_1}=\sum^{l_0}_{j=0}e^{2\pi i r f_{0,j}}P_{\gamma_1,\Z}|_{0,j}+\sum^m_{i=1}\sum^{l_i}_{j=0} (\zeta_{d_i})^{f_{i,j}}P_{\gamma_1,d_i}|_{i,j}\end{equation}

Where each $P_{\gamma_1,d_i}|_{i,j}$ acts on the $j$th copy of $\C[\Z_{d_i}]$ in each qudit. For a $(k+1)$-cell $\gamma_2$, adjacent to $\gamma_1$ that disagrees on the orientation and trivialization of some $k$-cell $\beta$ in $\partial \gamma$, the encoded state $\ket{a}$ corresponding to an element $a\in \pi_k(F)$ is transformed after changing trivialization to $(t^{i}_j)^*(a)$. After reinterpreting, $\ket{a}$ as $\ket{(t^{i}_j)^*(a)}$ for each boundary $k$-cell defined in other trivializations, we can apply the same $P$-type stabilizer as for $\gamma_1$.

The case where a maximal number of stabilizers are satisfied, i.e., there is the smallest number of violated $P^{obstr}_\gamma$ stabilizers, corresponds to a section that is maximally close to being extendable over the $k+1$-skeleton. 

To define this section up to homotopy on the $k$-skeleton, we must see what happens as we vary the section over the $(k-1)$-cells, as in the circle bundle example. For a $(k-1)$-cell $\alpha$, with section $\sigma_\alpha$ the sequence of maps that send a section over the $(k-1)$-skeleton to a homotopically equivalent section over the $(k-1)$-skeleton $\sigma_\alpha',\sigma_\alpha''\dots$ and back to $\sigma_\alpha$ trace a map from $S^k$ to the fiber. Suppose this were to construct a non-trivial element $g\in\pi_k(F)$. In that case, we have an action of the operator $ V _ {\ alpha, g} $ on the encoded data on the coboundary $k$-cells, where $V_{\alpha,g}$ is the $V$ operator associated with $g\in \pi_k(F)$. Therefore, if $\pi_k(F)=\Z^{l_0}\oplus^m_{i=1}\Z_{d_i}$, we must also include the stabilizers:

\begin{equation}V^{obstr}_\alpha=\sum^{l_0}_{j=0}V_{\alpha,\Z}|_{0,j}+\sum^m_{i=1}\sum^{l_i}_{j=0} V_{\alpha,d_i}|_{i,j}\end{equation}

As the $V_\alpha$ operators add a coboundary to the $k$-cells, the boundary of said $k$-cells is unchanged; so, this new operator commutes with the $P$-type operator as it did in the cohomological case. This commutation implies that adding these stabilizers does not change the manufactured erp cocycle; however, the section over the $k$-cells is now defined up to homotopy. We can thus conclude that an instance of the obstruction class is defined as the minimal error of this code.

Let $V_{\alpha,1}\dots V_{\alpha,{n_\alpha}}, P_{\gamma,1}\dots P_{\gamma,{n_\gamma}}$ correspond to each element in the sum for $V$ and $P$-type operators respectively. If we were to apply this construction using a physical choice of these stabilizers, as in section \ref{annoying:integer}, the associated Hamiltonian is:

\begin{equation}H_{Obstr}(v)=\sum_{\alpha\in T_{k-1}} \sum_{i=0}^{n_\alpha}(2-V_{\alpha,i}-V_{\alpha,i}^\dagger)(v)+\sum_{\gamma \in T_{k+1}} \sum_{j=0}^{n_\gamma}(1-P_{\gamma,j}-P_{\gamma,j}^\dagger)(v)\end{equation}

\subsection{Making the Hamiltonian Gapped and Finite} \label{homgroupfin}

We can see that if the homotopy group were not abelian and finitely generated, the above construction could not be done without losing substantial information. If we were to make this finitely generated version gapped, we must quotient $\pi_k(F)$ so that the non-torsion parts are finite. When doing so, the cohomology splits into a direct sum. By the universal coefficient theorem for cohomology\cite{hatcher2002algebraic}, letting $G=\pi_k(F)$, the code-space without the violated stabilizer after quotienting is:

\begin{equation}\C[H^k(B,G)\otimes \Z_m]\otimes \C[Tor_1(H^{k+1}(B,G),\Z_m)]\end{equation}

Furthermore, this is not canonically split, so to simplify the above result, if we can find an element $m$ which is coprime to the torsion part of $H^{k+1}(B, G)$ while not affecting the relevant torsion parts of $H^k(B, G)$, this would annihilate the $\C[Tor_1(H^{i+1}(B, G), \Z_m)]$ so that the code-space would be the characteristic class as an erp in $\C[H^k(B, G)\otimes \Z_m]$. We could also consider quotienting each group in the direct sum of $G$ separately, which would yield a similar outcome.

\phantomsection
\bibliographystyle{amsplain}
\bibliography{references}

\providecommand{\bysame}{\leavevmode\hbox to3em{\hrulefill}\thinspace}
\providecommand{\MR}{\relax\ifhmode\unskip\space\fi MR }
\providecommand{\MRhref}[2]{%
  \href{http://www.ams.org/mathscinet-getitem?mr=#1}{#2}
}
\providecommand{\href}[2]{#2}
\begin{thebibliography}{10}

\bibitem{dualcellcomplexpaper}
Tathagata Basak, \emph{Combinatorial cell complexes and poincar{\'e} duality}, Geometriae Dedicata \textbf{147} (2010), 357--387.

\bibitem{systole}
Marcel Berger, \emph{What is... a systole}, Notices of the AMS \textbf{55} (2008), no.~3.

\bibitem{BravyiHaah2011_EnergyLandscape3DSpin}
Sergey Bravyi and Jeongwan Haah, \emph{Energy landscape of 3d spin hamiltonians with topological order}, Physical Review Letters \textbf{107} (2011), no.~15, 150504.

\bibitem{BravyiHaah2013_QuantumSelfCorrection3DCubicCode}
\bysame, \emph{Quantum self-correction in the 3d cubic code model}, Physical Review Letters \textbf{111} (2013), no.~20.

\bibitem{hatcher2002algebraic}
Allen Hatcher, \emph{Algebraic topology}, Cambridge University Press, 2002.

\bibitem{HostensDehaeneDeMoor2005_StabilizerCliffordQudits}
Erik Hostens, Jeroen Dehaene, and Bart De~Moor, \emph{Stabilizer states and clifford operations for systems of arbitrary dimensions and modular arithmetic}, Physical Review A \textbf{71} (2005), no.~4.

\bibitem{kitaev2003fault}
A~Yu Kitaev, \emph{Fault-tolerant quantum computation by anyons}, Annals of physics \textbf{303} (2003), no.~1, 2--30.

\bibitem{kittel1980thermal}
Charles Kittel and Herbert Kroemer, \emph{Thermal physics}, Macmillan, 1980.

\bibitem{levin2005string}
Michael~A Levin and Xiao-Gang Wen, \emph{String-net condensation: A physical mechanism for topological phases}, Physical Review B—Condensed Matter and Materials Physics \textbf{71} (2005), no.~4, 045110.

\bibitem{milnor1974characteristic}
John~Willard Milnor and James~D Stasheff, \emph{Characteristic classes}, no.~76, Princeton university press, 1974.

\bibitem{alexanderduality}
Dale Rolfsen, \emph{Knots and links}, no. 346, American Mathematical Soc., 2003.

\bibitem{spanier1989algebraic}
Edwin~H Spanier, \emph{Algebraic topology}, Springer Science \& Business Media, 1989.

\end{thebibliography}
\end{document}